\documentclass{lmcs}
\pdfoutput=1
\usepackage[utf8]{inputenc}

\usepackage{lastpage}
\lmcsdoi{21}{3}{2}
\lmcsheading{}{\pageref{LastPage}}{}{}%
{Aug.~15,~2023}{Jul.~09,~2025}{}

\usepackage{listings}
\usepackage{hyperref}
\usepackage{caption, amsmath,amssymb, tikz, turnstile, amsthm, enumitem, verbatim, fancyhdr, ragged2e, latexsym, graphics, epsfig, epsf, rotating, subfigure, lastpage, hyperref}
\usetikzlibrary{arrows, automata, positioning}

\theoremstyle{definition}
\newtheorem{theorem}[thm]{Theorem}
\newtheorem{lemma}[thm]{Lemma}
\newtheorem{definition}[thm]{Definition}

\newcommand{\sderives}[1]{\sststile{#1}{}}

\newcommand{\dbox}[1]{\langle #1 \rangle}
\newcommand{\bbox}[1]{[ #1 ]}
\DeclareMathOperator{\nnf}{nnf}
\DeclareMathOperator{\config}{config}

\DeclareMathOperator{\?}{?}

\DeclareMathOperator{\mymodels}{models}

\DeclareMathOperator{\final}{final}
\DeclareMathOperator{\mypath}{hpath}
\DeclareMathOperator{\origin}{origin}
\DeclareMathOperator{\myindex}{index}

\newtheorem{corollary}[theorem]{Corollary}

\begin{document}

\title{Automata Linear Dynamic Logic on Finite Traces}

\author{Kevin W. Smith\lmcsorcid{0009-0006-3674-3289}}
\address{Computer Science Department, Rice University}
\email{kwsmith@fastmail.com}
\thanks{Based on MS Thesis of Kevin W.~Smith. Work supported in part by NSF grants IIS-1527668, CCF-1704883, IIS-1830549, DoD MURI grant N00014-20-1-2787, and an award from the Maryland Procurement Office.}

\author{Moshe Y. Vardi\lmcsorcid{0000-0002-0661-5773}}
\address{Computer Science Department, Rice University}
\email{vardi@rice.edu}

\begin{abstract}

  Linear Dynamic Logic on Finite Traces (LDL$_f$) is a temporal logic that extends propositional logic with regular expressions. Being equivalent in expressiveness to Monadic Second Order Logic, LDL$_f$ is more expressive than Linear Temporal Logic on Finite Traces (LTL$_f$), which is equivalent to Monadic First Order Logic, yet satisfiability checking for both LDL$_f$ and LTL$_f$ formulas is PSPACE-complete.

  Here we introduce \textit{Automata Linear Dynamic Logic on Finite Traces} (ALDL$_f$), which extends the paradigm of LDL$_f$ by using nondeterministic finite automata (NFA) in place of regular expressions and by providing for the direct expression of past modalities. NFA are as expressive as regular expressions, but have been shown to be exponentially more succinct. ALDL$_f$ is as expressive as LDL$_f$ and, as we show here,  satisfiability checking for ALDL$_f$ formulas is also PSPACE-complete, so the improved succinctness and past modalities come at no cost in complexity.
\end{abstract}

\maketitle

\section{Introduction}
\label{ch:introduction}

Temporal logics are widely used as specification languages for formal program verification and synthesis. They are also commonly used by the AI community for reasoning about actions and planning, such as expressing temporal constraints in task planning. One such logic, Linear Temporal Logic (LTL) \cite{pnueli77}, has been widely used in many contexts. It extends standard propositional logic with additional operators to express temporal constraints. It is simple to understand and easy to use effectively. Moreover, it has an attractive balance between formal expressiveness, i.e., the types of properties that can be specified by a formula, and the complexity of performing basic operations on formulas, such as satisfiability checking, which is a standard measure for determining the computational tractability of a formal logic.

A more recently developed temporal logic is Linear Dynamic Logic (LDL) \cite{ldlf}. LDL offers increased formal expressiveness, which makes it useful in cases when LTL's lack of expressiveness is an obstacle, which is not uncommon in practice \cite{ldlf}. LDL maintains a comparable ease of use to LTL. It combines propositional logic with the use of \textit{path expressions}, which are regular expressions used to express temporal constraints. Since most computer scientists are already familiar with regular expressions, it is easy for them to learn to use path expressions effectively.

The semantics of temporal logics may be defined over infinite traces or finite traces. The former allows for reasoning over an unbounded length of time, while the latter allows for reasoning over a finite time horizon. LTL was originally defined to operate over infinite traces \cite{pnueli77}, but more recent variants of LTL have been defined over finite traces and have proven to be useful in practice \cite{aalst,pevsic,ldlf}. Given the usefulness of finite-trace temporal logics and the increasing degree of interest in them, we limit our scope here to considering finite-trace variants of temporal logics and leave infinite-trace variants to future work.

A limitation of LDL$_f$ is that there exist some LDL$_f$ formulas for which there is currently no published well defined NFA construction. The only published LDL$_f$ construction, which is provided in \cite{ldlf} and used to check satisfiability, is not well defined for these formulas. There is a class of path expressions that if present in an LDL$_f$ formula will cause the automaton construction to fail. In particular, if a path expression contains a Kleene star operator whose operand is an expression that can be of length 0, then following the construction in \cite{ldlf} results in infinite recursion while building the transition function.

The main contribution of this work is the introduction of \textit{Automata Linear Dynamic Logic on Finite Traces} (ALDL$_f$), a variant of LDL$_f$, and showing that satisfiability for ALDL$_f$ is in PSPACE. ALDL$_f$ uses \textit{path automata}, which are nondeterministic finite automata (NFA) used to express temporal constraints. Using path automata rather than path expressions allows for an automaton construction that is valid for all ALDL$_f$ formulas. Moreover, by converting path expressions to path automata using classic algorithms such as Thompson's construction for converting regular expressions to NFA \cite{thompson}, the automaton construction for ALDL$_f$ can also be used for all LDL$_f$ formulas. Additionally, ALDL$_f$ provides for the direct expression of past constraints, which can be combined with present and future constraints within a single ALDL$_f$ formula. In comparison, LDL$_f$ can only directly express present and future constraints, while pure-past LDL$_f$ (PLDL$_f$) can only directly express present and past constraints \cite{pldlf}.

Additionally, ALDL$_f$ provides for the direct expression of past constraints, in contrast to LDL$_f$, which can only make direct claims about the present and future. As part of the means of establishing ALDL$_f$ satisfiability, we also introduce a novel variant of the \textit{two-way alternating automaton on finite words} (2AFW) that, in addition to conventional finite-run acceptance conditions, uses a B\"uchi-like acceptance condition on infinite runs to allow some infinite runs to be accepting.

ALDL$_f$ is equivalent in expressiveness to LDL$_f$, which is equivalent to Monadic Second-Order Logic, while LTL is equivalent to Monadic First-Order Logic. Since LTL satisfiability is PSPACE-complete and this work demonstrates that ALDL$_f$ satisfiability is in PSPACE, ALDL$_f$ provides greater expressiveness and succinctness than LTL$_f$ without increasing the complexity of satisfiability.
The motivation to add automata connectives to LDL$_f$ is two fold.
First, modern industrial property specification languages such as \emph{SystemVerilog Assertions} (SVA) allowed assertions to contain local variables, allowing the expression of local state, which is not allowed in LDL$_f$. Second, on the theoretical side, as finite-state automata are exponentially more succinct that regular expressions \cite{gruber2}, our investigation here explores the succinctness/complexity terrain for finite-horizon temporal logics, analogously to \cite{vardi-wolper}.

\subsection{Outline}

Our main goals are to provide a construction to translate a formula of ALDL$_f$ to a nondeterministic finite automaton (NFA) and to show that satisfiability checking of ALDL$_f$ formulas is in PSPACE.

We start with some preliminary definitions in Section 2. In Section 3, we examine related work in order to give historical context to ALDL$_f$. In Section 4, we formally define ALDL$_f$. This consists of defining the syntax of ALDL$_f$ formulas and the structure of ALDL$_f$ path automata and defining the semantics of ALDL$_f$ formulas, as well as defining the Fischer-Ladner closure for ALDL$_f$, which is a generalization of the concept of subformula that is appropriate for modal logics such as ALDL$_f$. In Section 5, we define the two-way alternating automaton on finite words (2AFW), which serves as an intermediate step in the ALDL$_f$ formula to NFA construction. In Section 6, we provide a construction that takes an ALDL$_f$ formula and produces an equivalent 2AFW and we prove the correctness of this construction. Finally, in Section 7, we show how to construct an equivalent NFA from a 2AFW and prove that this construction is correct. In Section 8, we show that ALDL$_f$ satisfiability is in PSPACE and we conclude by discussing the significance of this work.

\section{Preliminaries}
\label{ch:preliminaries}

In this section we introduce basic definitions that are necessary to understand the main body of this work.

An \textit{atomic proposition} is a declarative sentence that may be either true or false and cannot be divided such that its subcomponents may be either true or false. Atomic propositions are the most basic building blocks of logic and are usually represented as a variable.

A formula $\varphi$ of \textit{propositional logic} is generated by the following grammar:
\[ \varphi::= P \mid (\neg\varphi') \mid (\varphi_1 \wedge \varphi_2)\text{, where $AP$ is the set of atomic propositions and }P \in AP\]

A \textit{propositional interpretation} $\Pi$ is a set of atomic propositions that represents the truth values of
the most basic facts of the world: if an atomic proposition $P$ is in $\Pi$
then this represents that $P$ is true, otherwise that $P$ is false.

A propositional formula $\varphi$ is true within a propositional interpretation $\Pi$, that is, $\Pi \models \varphi$, iff one of the following holds:
\begin{enumerate}
    \item $\varphi = P$, where $P \in AP$, and $P \in \Pi$
    \item $\varphi = (\neg\psi)$ and $\Pi \not\models \psi$
    \item $\varphi = (\psi_1 \wedge \psi_2)$ and $\Pi \models \psi_1$, $\Pi \models \psi_2$
\end{enumerate}

The language defined by a propositional formula $\varphi$ is the set of propositional interpretations that make $\varphi$ true,
that is, $\mathcal{L}(\varphi) = \{ \Pi \mid \Pi \in 2^{AP} \text{ and } \Pi \models \varphi \}$.

Temporal logics with a linear-time model, such as Linear Temporal Logic (LTL) and Linear Dynamic Logic (LDL), generalize
propositional logic to facilitate formal reasoning about a world that can change over time. Such logics model time as a
discrete series of instances. A \textit{trace} $\pi = \pi_0 \pi_1 \ldots \pi_n$, where each $\pi_i \in 2^{AP}$, is a sequence of propositional interpretations. Intuitively, each element
of a trace represents what is true of the world at a particular instant. One can think of each element as a snapshot of the
current state of the world, and the trace itself can be thought of as a flip book that represents how the world changes over time.
In general a trace may be finite or infinite, but in the present work we are concerned only with finite traces.

A formula $\varphi$ of temporal logic may be evaluated against a
position $i$ within a trace $\pi$, where $i$ is a valid index in $\pi$.
If $\varphi$ is propositional then its truth value depends only on the content of the propositional interpretation
at position $i$. If $\varphi$ expresses a temporal constraint, then the truth value of $\varphi$ depends on the
content of propositional interpretations at position(s) other than $i$: if one or more propositional interpretations at positions less than $i$
are relevant then this constraint is a \textit{past modality}, and if one or more propositional interpretations at positions greater than $i$ are
relevant then this constraint is a \textit{future modality}.
The language defined by a formula of temporal logic $\varphi$ is the set of traces that make the formula true. That is,
$\mathcal{L}(\varphi) = \{ \pi \mid \pi \in (2^{AP})^* \text{ and } \pi \models \varphi \}$.

LTL$_f$, which is a variant of LTL defined for finite traces, is introduced in \cite{ldlf}. Its formulas are generated from the same grammar as standard LTL, which extends propositional logic with two additional operators, $X$ (``next'') and $\mathcal{U}$ (``until''):
\[ \varphi::= P \mid (\neg\varphi) \mid (\varphi \wedge \varphi) \mid (X\varphi) \mid (\varphi\ \mathcal{U}\ \varphi) \]

The semantics of LTL$_f$ are defined as follows. Let $\pi$ be a trace of finite length and $|\pi|$ be the length of $\pi$. The satisfaction of an LTL$_f$ formula $\varphi$ at time point $i$ on $\pi$, written as $\pi,i \models \varphi$, is inductively defined as follows:
\begin{enumerate}
    \item $\pi,i \models P$, $P \in AP$, iff $P \in \pi_i$
    \item $\pi,i \models \neg\varphi$ iff $\pi,i \not\models \varphi$
    \item $\pi,i \models \varphi_1 \wedge \varphi_2$ iff $\pi,i \models \varphi_1$ and $\pi,i \models \varphi_2$
    \item $\pi,i \models X\varphi$ iff $i < |\pi|$ and $\pi,i+1 \models \varphi$
    \item $\pi,i \models \varphi_1 $ $ \mathcal{U} $ $  \varphi_2$ iff for some $j$, $i \leq j < |\pi|$, for all $k$, $i \leq k < j$, $\pi,k \models \varphi_1$ and $\pi,j \models \varphi_2$
    \item $\pi \models \varphi$ iff $\pi,0 \models \varphi$
\end{enumerate}

Automata Linear Dynamic Logic on Finite Traces (ALDL$_f$) is a temporal logic whose formulas may make use of a construct called a \textit{path automaton}, appropriately defined later, to express temporal constraints. It is capable of expressing both past and future modalities.

A \textit{nondeterministic finite automaton} (NFA) is a 5-tuple  $(Q,\Sigma, \delta, q_0,F)$, where $Q$ is a finite set of states, $\Sigma$ is a finite input alphabet, $q_0 \in Q$ is the starting state, $F \subseteq Q$ is the set of accepting states, and $\delta : Q \times \Sigma \times Q$ is the transition relation. A word $w = c_0c_1 \ldots c_n$ is accepted by an NFA $A$ iff there is a sequence of states $s_0,s_1,\ldots,s_{n+1}$ such that $s_0 = q_0$, $(s_i, c_i, s_{i+1}) \in \delta$ for $0 \leq i \leq n$, and $s_{n+1} \in F$. The language of $A$, denoted $L(A)$ is the set of words it accepts, i.e., $L(A) = \{ w \mid w \text{ is accepted by A} \}$. If $U$ is an NFA with start state $q$ then $U_{q'}$ is the same automaton but with $q'$ as start state; that is, if $U = (Q,\Sigma,\delta,q,F)$ and $q' \in Q$, then $U_{q'} = (Q,\Sigma,\delta,q',F)$.

A \textit{two-way alternating automaton on finite words} (2AFW) is a 5-tuple $(Q, \Sigma, \delta, q_0, F)$, appropriately defined later in Section 5.

If $A = (Q,\Sigma,\delta,q_0,F)$ is an NFA or a 2AFW, we define $\final(q) = 1$ if $q \in F$, $\final(q) = 0$ otherwise.

We use the abbreviations $\varphi_1 \vee \varphi_2 \equiv \neg(\neg\varphi_1 \wedge \neg\varphi_2)$, $\varphi_1 \rightarrow \varphi_2 \equiv \neg \varphi_1 \vee \varphi_2$,
and $\varphi_1 \leftrightarrow \varphi_2 \equiv (\varphi_1 \rightarrow \varphi_2) \wedge (\varphi_2 \rightarrow \varphi_1)$.
If $\varphi$ is a formula then $AP(\varphi)$ denotes the set of atomic propositions that occur in $\varphi$. If $\varphi$ is a formula then $\nnf(\varphi)$ an equivalent formula in negation normal form.
We use $|\pi|$ to denote the length of $\pi$, $\pi_i$ to denote the $i$th element of $\pi$, and $[\pi]$ to denote the set $\{0,\ldots,|\pi|-1\}$ of positions in $\pi$.
The symbol \textbf{true} denotes a propositional formula that always evaluates to true and the symbol \textbf{false} denotes a propositional formula that always evaluates to false.

\section{Related Work}
\label{ch:relatedwork}
The use of temporal logic for use in formal program verification was first proposed by Pnueli in 1977 \cite{pnueli77}, using an adaptation of Tense Logic, which was first developed by Prior in 1957 to facilitate precise reasoning about time for use in solving philosophical problems \cite{prior57}. Pnueli's adaptation, which came to be known as Linear Temporal Logic (LTL), was tailored to have the necessary expressiveness to specify temporal dependencies in programs. Pnueli's LTL was originally used over infinite traces. More recently LTL interpreted over finite traces (LTL$_f$) was proposed in \cite{ldlf} and shown to be a useful specification language \cite{aalst,bienvenu,gabaldon,pevsic,wilke}.

Regular expressions over propositional formulas (RE$_f$), where a propositional formula $\varphi$ serves as a representation of the set of propositional interpretations that satisfy $\varphi$, are a more expressive alternative to LTL$_f$, with regular expressions being equivalent in expressiveness to monadic second-order logic and LTL$_f$ being equivalent in expressiveness to monadic first-order logic. Nevertheless, RE$_f$ lacks direct constructs for negation and conjunction. These can be added by allowing complementation and intersection of regular expressions, but this results in nonemptiness checking being of nonelementary complexity even for star-free regular expressions.

Linear Dynamic Logic on Finite Traces (LDL$_f$) is a temporal logic that extends propositional logic with operators that use regular expressions to express temporal constraints \cite{ldlf}. It has the formal expressiveness of RE$_f$, while having constructs for negation and conjunction. Notably, satisfiability checking for LDL$_f$ formulas is PSPACE-complete, the same as for LTL$_f$. Thus, LDL$_f$ is more expressive than LTL$_f$ at no additional cost \cite{ldlf}. Pure-Past Linear Dynamic Logic on Finite Traces (PLDL$_f$) is a variant of LDL$_f$ that allows for the expression of present and past temporal constraints \cite{pldlf}, as opposed to LDL$_f$, which can express present and future temporal constraints.

\section{\texorpdfstring{Introducing ALDL$_f$}{Introducing ALDLf}}

\label{ch:introducingaldlf}
Automata Linear Dynamic Logic on Finite Traces (ALDL$_f$) is a variant of LDL$_f$ that extends propositional logic with operators that use nondeterministic finite automata (NFA) to express temporal constraints. It also provides direct constructs for past modalities. The use of NFAs rather than regular expressions offers two advantages. The first advantage is that NFAs are exponentially more succinct than regular expressions \cite{gruber1,gruber2}. Satisfiability checking of ALDL$_f$ formulas is PSPACE-complete, so this succinctness comes at no cost in complexity. The second advantage is that sometimes NFAs can be a more convenient form of expression than regular expressions. It is sometimes easier in practice to describe behavior using a state-based formalism, as can be observed by the the use of such formalism in industrial language. For example, SystemVerilog Assertions (SVA) uses local variables \cite{sva}, which are a form of state.

Moreover, ALDL$_f$ and its satisfiability-checking procedure can be used for LDL$_f$, by converting an LDL$_f$ formula to an equivalent ALDL$_f$ formula. This can be done by converting each occurrence of regular expression in the LDL$_f$ formula to an equivalent NFA using a standard algorithm, e.g., Thompson's construction \cite{thompson}. Having an alternative method to perform LDL$_f$ satisfiability checking to the one provided in \cite{ldlf} is advantageous because there are syntactially valid LDL$_f$ formulas for which the construction of the transition function for the AFW provided in \cite{ldlf} may not terminate. A simple example is $\varphi = \dbox{p^{**}}\textbf{false}$, where $p$ is an atomic proposition. It is easy to see that following the AFW transition function construction provided in Theorem 15 of \cite{ldlf} on $\varphi$ results in a loop. Nevertheless, converting $\varphi$ to an ALDL$_f$ formula and performing the satisfiability checking procedure provided here is well founded.

We are now ready to formally define ALDL$_f$. First we define its syntax, then its semantics, and finally we provide the Fischer-Ladner Closure for ALDL$_f$, which is a generalization of the concept of subformula that will be crucial in future sections.

\subsection{\texorpdfstring{ALDL$_f$ Syntax}{ALDLf Syntax}}
The first step in defining ALDL$_f$ is to define its syntax. We define two objects, the syntax of an ALDL$_f$ formula and the structure of a path automaton.
We define the objects by mutual recursion because an ALDL$_f$ formula may have path automata as components
and a path automaton may have ALDL$_f$ formulas in its alphabet in the form of \textit{tests}, which are a special construct
that syntactically consists of an ALDL$_f$ formula followed by a question mark.

  Let $AP$ be the set of atomic propositions, $PForm$ be the set of propositional formulas over $AP$, and $PForm^- = \{ \zeta^- \mid \zeta \in PForm \}$ (note that $\zeta^-$ is the symbol $\zeta$ with a marker; it is not negation). The alphabet of a path automaton includes elements from
    $PForm$ and $PForm^-$, which are used to express temporal modalities. The elements of
    $PForm$ are used to express future modalities, while the elements of $PForm^-$ are used to
    express past modalities. As mentioned earlier, a path automaton's alphabet may contain ALDL$_f$ formulas in the form
    of tests, which allow the automaton to make a transition without consuming input so long as the test is
    satisfied, which is explained further in the semantics section.
    A question mark following a formula indicates that the formula is a test.

    \begin{definition}\label{def:aldlf_mutual} We define syntax of an ALDL$_f$ formula and the structure of a path
        automaton by mutual recursion:

    \textbf{(1.1):} An ALDL$_f$ formula $\varphi$ is generated by the following grammar:

  $\varphi ::= P \mid (\neg \varphi) \mid (\varphi_1 \wedge \varphi_2) \mid (\varphi_1 \vee \varphi_2)
  \mid (\dbox{\alpha} \varphi) \mid (\bbox{\alpha} \varphi)$, where $P \in AP$ and $\alpha$ is a path automaton:

    \textbf{(1.2):} A path automaton is an NFA $(R,T,\Delta,r,G)$, where
    \begin{itemize}
        \item $R$ is a finite set of states
        \item $T$ is the alphabet and consists of a finite subset of $PForm \, \cup \, PForm^- \, \cup \, \{ \psi \? \mid \psi \text{ is an} \allowbreak
        \text{ALDL$_f$ formula} \}$
        \item $r \in R$ is the start state
        \item $G \subseteq R$ is the set of accepting states
        \item $\Delta \subseteq R \times T \times R$ is the transition relation
        \item Each state must occur in the transition relation, i.e., for all $s \in R$, there exists a $d \in \Delta$ such that either $d = (s, \tau, s')$ or $d = (s', \tau, s)$, for some $s' \in R$ and $\tau \in T$.
    \end{itemize}

\end{definition}

We now define the \emph{size} of a path automaton:
\begin{definition}
The size of a path automaton $\alpha =(R, T, \Delta, r, G)$, denoted $|\alpha|$, is $2|\Delta|\log(|R|)$ + $\sum\limits_{\psi\? \in T}^{} |\psi|$.
\end{definition}

We use $2|\Delta|\log (|R|)$,  as each transition has both a start and an end state and we need $\log(|R|)$ to denote $|R|$ states. (Note that the cardinality of $R$ is at most the cardinality of $\Delta$ because of the requirement that each state of the automaton occurs in the transition relation.)

Now we can describe the \emph{size} of formulas. The size of a formula is (roughly) how many characters it takes to write down, so each occurrence of parenthesis, connective, $<$, $>$, [, and ] contribute 1 to the size. Formally: 
\begin{definition}
  The size of an ALDL$_f$ formula $\varphi$, denoted $|\varphi|$, is defined through structural induction:

  \textbf{Base case:} $\varphi = P$, where $P$ is an atomic proposition. Then $|\varphi| = 1$.

  \textbf{Inductive step:}

  \textit{Case (1):} $\varphi = (\neg\varphi')$, where $\varphi'$ is an ALDL$_f$ formula. Then $|\varphi| = 3 + |\varphi'|$.

  \textit{Case (2):} $\varphi = (\varphi_1 \wedge \varphi_2)$. Then $|\varphi| = 3 + |\varphi_1| + |\varphi_2|$.

  \textit{Case (3):} $\varphi = (\varphi_1 \vee \varphi_2)$. Then $|\varphi| = 3 + |\varphi_1| + |\varphi_2|$.

  \textit{Case (4):} $\varphi = (\dbox{\alpha}\varphi')$. Then $|\varphi| = 4 + |\alpha| + |\varphi'|$.

  \textit{Case (4):} $\varphi = (\bbox{\alpha}\varphi'$). Then $|\varphi| = 4 + |\alpha| + |\varphi'|$.

  \end{definition}

  Note that the size of a path automata depends on the size of a formula and the size of a formula depends on the size of a path automaton. No cycles are allowed and this can be done carefully by mutual induction.

\subsection{\texorpdfstring{ALDL$_f$ Semantics}{ALDLf Semantics}}

Having defined the syntax of ALDL$_f$, we proceed by defining the meaning of ALDL$_f$ formulas. We introduce a ternary relation for satisfaction of ALDL$_f$ formulas and a quaternary relation for
satisfaction of path automata and define them by mutual recursion.

\begin{definition}\label{def:aldlf_semantics}
If the conditions in 4.1 below hold, then
an ALDL$_f$ formula is true within a trace $\pi \in (2^{AP})^+$ at position $i$,
    where $0 \leq i < |\pi|$, (denoted $\pi,i \models \varphi$),
    and an ALDL$_f$ path automaton $(R,T,\Delta,r,G)$ is satisfied by the subsequence of $\pi$ from position $i$ to position $j$ (denoted
    $\pi,i,j \models (R,T,\Delta,r,G)$ if the conditions in \ref{def:aldlf_semantics}.2 hold:

    \textbf{(\ref{def:aldlf_semantics}.1)}:
    \begin{enumerate}[label=A\arabic*.]
    \item $\pi,i \models P$ iff $P \in \pi_i$ ($P \in AP$)
    \item $\pi,i \models (\neg\varphi$) iff $\pi,i \not\models \varphi$
    \item $\pi,i \models (\varphi_1 \wedge \varphi_2)$ iff $\pi,i \models \varphi_1$ and $\pi,i \models \varphi_2$
    \item $\pi,i \models (\varphi_1 \vee \varphi_2)$ iff $\pi,i \models \varphi_1$ or $\pi,i \models \varphi_2$
    \item $\pi,i \models (\dbox{(R,T,\Delta,r,G)}\varphi)$ iff there exists $j$ such that $0 \leq j < |\pi|$, $\pi,i,j \models (R,T,\Delta,r,G)$, and $\pi,j \models \varphi$
    \item $\pi,i \models (\bbox{(R,T,\Delta,r,G)}\varphi)$ iff for all $0 \leq j < |\pi|$ such that $\pi,i,j \models (R,T,\Delta,r,G)$ we have that $\pi,j \models \varphi$
    \end{enumerate}

    \textbf{(\ref{def:aldlf_semantics}.2)}:

$\pi,i,j \models (R,T,\Delta,r,G)$ iff $i = j$ and $r \in G$ or there exists a finite sequence $(r_0,k_0),\allowbreak (r_1,k_1),\ldots,(r_n,k_n)$ over $(R \times [\pi])$,  such that $r_0 = r$, $k_0 = i$, $r_n \in G$, $k_n = j$, and for all $0 \leq m < n$ we have that one of the following holds:

\begin{enumerate}
    \item $k_{m+1} = k_m + 1$ and there is some $(r_m,\zeta,r_{m+1}) \in \Delta$ such that $\pi,k_m \models \zeta$
    \item $k_{m+1} = k_m - 1$ and there is some $(r_m,\zeta^-,r_{m+1}) \in \Delta$ such that $\pi,k_m \models \zeta$
    \item $k_{m+1} = k_m$ and there is some $(r_m,\psi\?,r_{m+1}) \in \Delta$ such that $\pi,k_m \models \psi$
\end{enumerate}

\end{definition}
Intuitively, a path automaton divides a trace into a prefix and a remainder (there may be more than one such division). In the case of a formula of the form $\dbox{U}\varphi$, a trace $\pi$ will satisfy this formula just in case that\textit{there exists} a prefix of $\pi$ that is in the language of $U$ such that the remainder of $\pi$ is in the language of $\varphi$. In the case of a formula of the form $\bbox{U}\varphi$, a trace $\pi$ will satisfy this formula just in case that \textit{for every} prefix of $\pi$ that is in the language of $U$ the remainder of $\pi$ is in the language of $\varphi$. Note that $\dbox{U}\varphi \equiv \neg \bbox{U}\neg \varphi$
and $\bbox{U}\varphi \equiv \neg \dbox{U}\neg\varphi$.

A trace $\pi$ is in the language of an ALDL$_f$ formula $\varphi$, denoted $\pi \models
\varphi$, if $\pi,0 \models \varphi$. If $\zeta$ is a propositional ALDL$_f$ formula, the notation $\pi_i \models \zeta$ may be used as a synonym for $\pi,i \models \zeta$, in order to emphasize the fact that the truth value of $\zeta$ at $i$ depends only on the contents of $\pi_i$. An ALDL$_f$ formula $\varphi$ is in \textit{negation normal form}
if all instances of the negation operator ($\neg$) that occur in $\varphi$, including those that occur in tests that
are included in a path automaton's alphabet, are applied
to atomic propositions. An ALDL$_f$ formula can be converted to negation normal form by replacing all instances of
abbreviations with their definitions, converting all instances of $\neg\bbox{U}\varphi$ to $\dbox{U}\neg\varphi$ and
all instances of $\neg\dbox{U}\varphi$ to $\bbox{U}\neg\varphi$, applying DeMorgan's rules to make all
negations appear as inwardly nested as possible, and then eliminating all double negations. If $\varphi$ is a formula, then nnf($\varphi$) is the result of performing this procedure on $\varphi$.

\subsection{\texorpdfstring{Fischer-Ladner Closure for ALDL$_f$}{Fischer-Ladner Closure for ALDLf}}

As we have just seen, the satisfaction relation for ALDL$_f$ is defined recursively, i.e., the satisfaction of a formula is defined in terms of the satisfaction of other formulas with the exception of some base cases. This is similar to how satisfaction is defined for LTL. Our strategy to translate ALDL$_f$ formulas to automata will be to use states of the automaton to represent formulas that are ``reachable'' during the recursive evaluation of satisfaction, which is similar to how a similar translation is done for LTL. It is necessary to know in advance which formulas may need to be evaluated in order to construct the automaton. In the case of an LTL formula $\varphi$, the subformulas of $\varphi$ are sufficient. But in the case of modal logics such as ALDL$_f$,
the concept of subformula is too weak and so we need to use something more broad. The \textit{Fischer-Ladner closure}, introduced in \cite{fischer79} for use in Propositional Dynamic Logic (PDL), is a generalization of subformulas for modal logics. Here we define the Fischer-Ladner closure for ALDL$_f$.

\begin{definition}
    The Fischer-Ladner closure for an ALDL$_f$ formula $\varphi$, denoted $CL(\varphi)$, is the smallest set for which
    the following conditions hold:
\begin{enumerate}
    \item $\varphi$ is in $CL(\varphi)$
    \item If $(\neg\varphi') \in CL(\varphi)$ and $\varphi'$ is not of the form $\neg\varphi''$ then $\varphi' \in CL(\varphi)$
    \item If $(\varphi_1 \wedge \varphi_2) \in CL(\varphi)$ then $\varphi_1,\varphi_2 \in CL(\varphi)$
    \item If $(\varphi_1 \vee \varphi_2) \in CL(\varphi)$ then $\varphi_1,\varphi_2 \in CL(\varphi)$
    \item If $(\dbox{U}\varphi') \in CL(\varphi)$,
    where $U = (R,T,\Delta,r,G)$, then
    \begin{enumerate}[label=(\roman*)]
        \item $\varphi' \in CL(\varphi)$
        \item For all $r' \in R$, $\dbox{(R,T,\Delta,r',G)}\varphi' \in CL(\varphi)$
        \item For all $\psi\? \in T$, $\psi \in CL(\varphi)$
    \end{enumerate}
    \item If $(\bbox{U}\varphi') \in CL(\varphi)$, where $U = (R,T,\Delta,r,G)$, then
    \begin{enumerate}[label=(\roman*)]
        \item $\varphi' \in CL(\varphi)$
        \item For all $r' \in R$, $\bbox{(R,T,\Delta,r',G)}\varphi' \in CL(\varphi)$
        \item For all $\psi\? \in T$, $\nnf(\neg\psi) \in CL(\varphi)$
    \end{enumerate}

\end{enumerate}
\end{definition}

\begin{lemma}

  Let $\varphi$ be an ALDL$_f$ formula. Then the cardinality of $CL(\varphi)$, denoted $|CL(\varphi)|$, is less than or equal to the size of $\varphi$.

\end{lemma}

\pagebreak[5]
\begin{proof}
  We prove the statement by structural induction.

  \textbf{Base case:} $\varphi = P$, where $P$ is an atomic proposition. Now, $CL(P) = \{P\}$, so $|CL(\varphi)| = 1$, while $|\varphi| = 1$, so $|CL(\varphi)| = |\varphi|$.

  \textbf{Inductive step:}

  \textit{Case (1):} $\varphi = (\neg\varphi')$. Then $CL(\varphi) = \{\neg\varphi'\} \cup CL(\varphi')$, so $|CL(\varphi)| = 1 + |CL(\varphi')|$. Now, $|\varphi| = 3 + |\varphi'|$. By the inductive hypothesis we have that $|CL(\varphi')| \leq |\varphi'|$. So $|CL(\varphi)| < |\varphi|$.

  \textit{Case (2):} $\varphi = (\varphi_1 \wedge \varphi_2)$. So $CL(\varphi) = \{\varphi\} \cup CL(\varphi_1) \cup CL(\varphi_2)$. So $|CL(\varphi)| = 1 + |CL(\varphi_1)| + |CL(\varphi_2)|$. Now, $|\varphi| = 3 + |\varphi_1| + |\varphi_2|$. By the inductive hypothesis we have that $|CL(\varphi_1)| \leq |\varphi_1|$ and $|CL(\varphi_2) \leq |\varphi_2|$, so $|CL(\varphi)| < |\varphi|$.

  \textit{Case (3):} $\varphi = (\varphi_1 \vee \varphi_2)$. So $CL(\varphi) = \{\varphi\} \cup CL(\varphi_1) \cup CL(\varphi_2)$. So $|CL(\varphi)| = 1 + |CL(\varphi_1)| + |CL(\varphi_2)|$. Now, $|\varphi| = 3 + |\varphi_1| + |\varphi_2|$. By the inductive hypothesis we have that $|CL(\varphi_1)| \leq |\varphi_1|$ and $|CL(\varphi_2)| \leq |\varphi_2|$, so $|CL(\varphi)| < |\varphi|$.

  \textit{Case (4):} $\varphi = \dbox{(R,T,\Delta,r,G)}\varphi'$. Let $CL(\varphi) = \{\varphi\} \cup \{(R,T,\Delta,r',G)\varphi' \mid r' \in R \} \cup \{ \psi \mid \psi\? \in T\}$. So $|CL(\varphi)| = 1 + |R| + |\bigcup\limits_{\psi\? \in T}CL(\psi)|$. Now, $|\varphi| = 4 + 2|\Delta|\log(|R|) + \sum\limits_{\psi\? \in T}|\psi|$. Because of the requirement that each state in $R$ must occur in $\Delta$ we have that $|R| \leq |\Delta|$. By the inductive hypothesis we have that for all $\psi\? \in T$, $|CL(\psi)| \leq |\psi|$. So $|CL(\varphi)| < |\varphi|$.

  \textit{Case (5):} $\varphi = \bbox{(R,T,\Delta,r,G)\varphi'}$. The reasoning for this case is the same as that of Case (4).
\end{proof}

\subsection{Example}

Here we provide an ALDL$_f$ formula that uses both future and past modalities as well as both existential and universal modal operators to define an easy-to-understand property.

We start by defining the path automata the formula uses. The first path automaton, $\alpha_1$, simply reads backward once and then accepts. Formally, $\alpha_1 = (R_{\alpha_1}, T_{\alpha_1}, r_{\alpha_1}, G_{\alpha_1}, \Delta_{\alpha_1})$, where

\begin{itemize}
  \item $R_{\alpha_1} = \{ p_0, p_1 \}$
  \item $T_{\alpha_1} = \{ \textbf{true}^- \}$
  \item $r_{\alpha_1} = p_0$
  \item $G_{\alpha_1} = \{ p_1  \}$
  \item $\Delta_{\alpha_1} = \{ p_0, \textbf{true}^-, p_1 \}$
\end{itemize}

The second path automaton, $\alpha_2$, reads backward through the trace until it finds a position at which the proposition $a$ is true and then accepts. Formally, $\alpha_2 = (R_{\alpha_2}, T_{\alpha_2}, r_{\alpha_2}, G_{\alpha_2}, \Delta_{\alpha_2})$, where

\begin{itemize}
  \item $R_{\alpha_2} = \{ q_0, q_1 \}$
  \item $T_{\alpha_2} = \{ a\?, \textbf{true}^- \}$
  \item $r_{\alpha_2} = q_0$
  \item $G_{\alpha_2} = \{ q_1 \}$
  \item $\Delta_{\alpha_2} = \{ (q_0, \textbf{true}^-, q_0), (q_0, a\?, q_1) \}$
\end{itemize}

Finally, the third path automaton, $\alpha_3$, reads forward through the trace until it finds a position at which the proposition $a$ is true and then accepts. Formally, $\alpha_3 = (R_{\alpha_3}, T_{\alpha_3}, r_{\alpha_3}, G_{\alpha_3}, \allowbreak\Delta_{\alpha_3})$, where

\begin{itemize}
  \item $R_{\alpha_3} = \{ s_0, s_1 \}$
  \item $T_{\alpha_3} = \{ a\?, \textbf{true} \}$
  \item $r_{\alpha_3} = s_0$
  \item $G_{\alpha_3} = \{ s_1 \}$
  \item $\Delta_{\alpha_1} = \{ (s_0, \textbf{true}, s_0), (s_0, a\?, s_1) \}$
\end{itemize}

We can use these path automata to create a formula that defines the language of traces in which $a$ being true anywhere in the trace implies that $a$ is true at the first position in the trace: $ \varphi = \bbox{\alpha_3} \dbox{\alpha_2} \bbox{\alpha_1} \textbf{false}$. For the purpose of explanation, we break down $\varphi$ into these subformulas:

\begin{itemize}
  \item $\varphi_1 = \bbox{\alpha_1}\textbf{false}$

  \item $\varphi_2 = \dbox{\alpha_2}\varphi_1$

  \item $\varphi_3 = \bbox{\alpha_3} \varphi_2$
\end{itemize}

In $\varphi_1$, $\bbox{\alpha_1}$ imposes the obligation to satisfy the unsatisfiable subformula \textbf{false} at the position in the trace immediately preceding the position at which $\varphi_1$ is evaluated, with the one exception being when $\varphi_1$ is evaluated at the first element in the trace. When $\varphi_1$ is evaluated at the first position in the trace it is vacuously true because there is no $j \geq 0$ such that $\pi,i,j \models \alpha_1$. So $\varphi_1$ can be used to find the beginning of the trace.

In $\varphi_2$, $\dbox{\alpha_2}$ looks backward through the trace to find an element at which $a$ is true. Because $\varphi_2$ uses the existential modal operator, it must find a position in the trace where $a$ is true and $\varphi_1$ is satisfied. Thus, for $\varphi_2$ to be true anywhere in the trace, it must be the case that $a$ is true at the first position in the trace.

In $\varphi_3$, $\bbox{\alpha_3}$ looks forward through the trace to find positions at which $a$ is true. Because $\varphi_3$ uses the universal modal operator, all such positions trigger the obligation to satisfy $\varphi_2$. Thus, if $a$ is true anywhere in the trace, then it must be true that $a$ is true at the first position of the trace.

\section{Two-way Alternating Automata on Finite Words (2AFW)}

\label{ch:introducing2afw}

In order to serve as an intermediate step in the translation of formulas of ALDL$_f$ to equivalent
NFAs, we use a \textit{two-way alternating automaton on finite words} (2AFW). A useful history and
taxonomy of two-way alternating automata is provided by Kapoutsis and Zakzok \cite{kapoutsis}. The 2AFW we
define here is similar to that of Geffert and Okhotin \cite{geffert}. The translation of ALDL$_f$ formula to NFA will first convert
an ALDL$_f$ formula to an equivalent 2AFW and then convert the result to an equivalent NFA.

An automaton is \textit{two-way} if its transition function allows the read position of the
input to not only move forward, but to stay in place or move backward as well. The direction
that the read position should move are represented by one of $\{-1,0,1\}$; -1 indicates
that the read position should move backward, 0 that it should stay in place, and 1 that it
should move forward.

In a 2AFW, the transitions combine nondeterminism and universality to allow nondeterministic
transitions to sets of state-direction pairs. Because of universal transitions a run of a 2AFW
is a tree, and nondeterminism allows for the possibility of there being more than one run of
a 2AFW on the same input. The combination of nondeterminism and universality is achieved by
representing each transition as a formula of propositional logic with state-direction pairs
as propositional atoms; the set of possible transitions is the set of state-direction pairs that
satisfy the formula. Disjunctions in a formula allow for nondeterminism, while conjunctions
allow for universality.

For a given set, we define the relevant formulas and their satisfaction condition:

    \begin{definition}\label{def:positive_boolean}
If $X$ is a set, then $B^+(X)$ is the set of positive Boolean formulas over the elements of
    $X$, that is, $B^+(X)$ is the closure of $X$ under disjunction and conjunction; note
    the lack of negation. Additionally, for all $X$, $B^+(X)$ contains the formulas
    \textbf{true} and \textbf{false}.  A set
$M \subseteq X$ \textit{satisfies} a formula $\theta \in B^+(X)$ if the truth
assignment that assigns true to the members of $M$ and false to the members of
$X - M$ satisfies $\theta$.
\end{definition}
We also use $\mymodels(\theta) = \{M \mid M \models \theta \}$.

We are now ready to define the 2AFW:

    \begin{definition}\label{def:2afw_simple}
A \textit{two-way alternating automaton on finite words} is a tuple $A = (Q, \Sigma, \delta, \allowbreak q_0,
F)$, where $Q$ is a finite nonempty set of states, $\Sigma$ is a finite nonempty alphabet, $q_0
\in Q$ is the initial state, $F \subseteq Q$ is a set of accepting states, and $\delta :
Q \times \Sigma \rightarrow B^+(Q \times \{-1, 0, 1\})$ is the transition function.
    \end{definition}

Runs of a 2AFW are trees whose nodes are labeled to indicate the state of the automaton and
    the current read position in the input. We call such labels \textit{configurations}: A \textit{configuration} of $A$ is a member of $(Q \times \mathbb{N}) \cup \{Accept\}$. The $Accept$
configuration represents a successful computation while configurations of the other form represent
computations that are still in progress.  If a
configuration $c$ is not $Accept$ then we call $c$'s state component the state and
we call $c$'s integer component the position. The state component keeps track of which state the automaton
is in during the run while the position component keeps track of where the input is being read from during the run.

In a 2AFW transition, the read position of the input is either unchanged or it moves forward or
backward exactly one position. We formalize this condition with the following relation:

    \begin{definition}\label{def:sderives}
We define the binary relation $\sderives{\tau}$ over configurations, where $\tau \in Q \times
\{-1,0,1\}$, as follows:

$(q,i) \sderives{(p,-1)} (p,i-1)$ for $i > 0$

$(q,i) \sderives{(p,0)} (p,i)$

$(q,i) \sderives{(p,1)} (p, i+1)$

    \end{definition}

    We can now define runs on a 2AFW:

\begin{definition}\label{def:run}
A $run$ of $A$ on a finite, nonempty word $w = w_0w_1 \ldots w_n$ from position $k$, where $n \in \mathbb{N}$ and $0 \leq k \leq n$, is a configuration-labeled tree $\rho$ such that the root of $\rho$ has configuration $(q_0, k)$ and if $x$ is a node of $\rho$ and has a configuration of the form $(q,i)$ then $0 \leq i \leq n$ and the following holds. Let $\theta = \delta(q,w_i)$. If $\theta = \textbf{false}$ then $x$ has no children.  If $\theta = \textbf{true}$ then $x$ has one child whose configuration is $\textit{Accept}$. Otherwise, for some $M \in \mymodels(\theta)$, for all $\tau \in M$, $x$ has a child $y$ such that $(q,i) \sderives{\tau} \config (y)$, where $\config (y)$ denotes the configuration of $y$, and the number of $x$'s children is equal to $|M|$. If a run is described without including a position, the starting position is presumed to be 0.

A branch $b$ of $\rho$ is an \textit{accepting branch} if it is a finite branch that ends in a leaf labeled \textit{Accept} or it is an infinite branch and there is some accepting state $f \in F$ that occurs in infinitely many configuration labels in $b$.

A run $\rho$ is \textit{accepting} if all of its branches are accepting branches.
\end{definition}

An automaton $A$ accepts a word $w$ from position $i$ if there is an accepting run of $A$ on $w$ from $i$.
$A$ defines the language $\mathcal{L}(A) = \{ (w,i) \mid A \text{ accepts } w \text{ from } i \}$.

\section{\texorpdfstring{From ALDL$_f$ Formulas to 2AFWs}{From ALDLf Formulas to 2AFWs}}
\label{ch:2afwtoaldlf}

Having defined the 2AFW, we are now ready to begin the first part of the ALDL$_f$ formula to NFW translation. The construction starts by converting from an ALDL$_f$ formula $\varphi$  in negation normal form to a 2AFW $A_\varphi$ that defines the same language, i.e., $\mathcal{L}(\varphi) = \mathcal{L}(A_\varphi)$.
Definition \ref{def:aldlf_to_2afw} provides the formal definition. We provide an intuitive explanation here.

The states of $A_\varphi$ are the Fischer-Ladner closure of
$\varphi$. The nodes of a run $\rho$ of $A_\varphi$ on a trace $\pi$
are labeled with states and positions (with
the exception of the special case of leaf nodes labeled with configuration \textit{Accept};
if a node $v$ of $\rho$ has state $\psi$ and position $i$ it is useful to think of the
subtree with $v$ as root to be a run of $A_\psi$ on $\pi$ from $i$.

Naturally, the start state of $A_\varphi$ is $\varphi$.

The alphabet of $A_\varphi$ is $2^{AP(\varphi)}$, i.e., elements of a trace. The set of accepting states consists of all states of the form $\bbox{U}\varphi$; this may seem counter-intuitive but this choice is explained below.

The transition function ensures that runs of $A_\varphi$ have the same recursive structure as the semantics of $\varphi$. Rule 1 corresponds to rule A1 of the semantics. Rule 2 corresponds to rule A2; since $\varphi$ is assumed to be in negation normal form we need only handle
negations of atomic propositions. Rules 3 and 4 correspond to rules A3 and A4, respectively.

Rule 5 corresponds to rule A5. If a path automaton gets ``stuck'', i.e., it is not in an accepting state and cannot make a transition, then all of the disjunctions evaluate to \textbf{false}. This is desirable because a stuck path automaton in this context represents that the first condition of rule A5 does not hold. Similarly, if the path automaton gets trapped in a loop and never reaches an accepting state, this also represents that the first condition of rule A5 does not hold. In this latter case, it results in a run with an infinite branch in which each node of the branch is labeled with a state of the form $\dbox{U}\varphi$. Since there are no states of the form $\bbox{U}\varphi$, this is a rejecting branch, as desired.

Since states of the form $\dbox{U}\varphi$ ``demand'' reaching an appropriate position that satisfies $\varphi$, getting ``stuck'' in such a state means that the ``demand'' has not been met. Thus, states of the form $\dbox{U}\varphi$ are not accepting states.

Rule 6 corresponds to rule A6. If a path automaton gets stuck, then all of the conjunctions will be empty and evaluate to \textbf{true}. This is desirable because it represents the vacuous satisfaction of rule A6. Similarly, if the path automaton becomes trapped in a loop and never reaches an accepting state, this also represents the vacuous satisfaction of rule A6. In a run, a trapped path automaton looping forever results in an infinite branch. Since every node on such a branch is labeled with states of the form $\bbox{U}\varphi$, it is an accepting branch, as desired,which is why we defined states of the form $\bbox{U}\varphi$ to be accepting.

\subsection{Construction}

We now formally define the construction:

\begin{definition}\label{def:aldlf_to_2afw}
Given an ALDL$_f$ formula $\varphi$ in negation normal form, we define a two-way alternating automaton $A_\varphi =
(Q,\Sigma,\delta,q_0,F)$, where $Q = CL(\varphi)$, $\Sigma = 2^{AP(\varphi)}$, $q_0 = \varphi$, $F
    \subseteq CL(\varphi)$ contains all formulas of the form $\bbox{U}\psi$, and
$\delta : Q \times \Sigma \rightarrow B^+(Q \times \{-1,0,1\})$ is defined for a propositional interpretation $\Pi$ by the following rules:

\begin{enumerate}
\item $\delta(P,\Pi) = \begin{cases}
  \textbf{true}, & \text{if } \Pi \models P \quad (P \in AP)\\
  \textbf{false}, & \text{otherwise}
\end{cases}$

\item $\delta(\neg P, \Pi) = \begin{cases}
 \textbf{true}, & \text{if } \Pi \not\models P \quad (P \in AP)\\
 \textbf{false}, & \text{otherwise}
  \end{cases}$

\item $\delta(\varphi_1 \wedge \varphi_2, \Pi) = (\varphi_1, 0) \wedge (\varphi_2, 0)$
\item $\delta(\varphi_1 \vee \varphi_2, \Pi) = (\varphi_1, 0) \vee (\varphi_2, 0)$
\\[3mm]
Note that in 3 and 4, the index of the word does not change and $\Pi$ will be read at a later stage.

Let the path automaton $U = (R,T,\Delta,r,G)$.

\item $\delta(\dbox{U}\varphi,\Pi) = \begin{cases}

\quad (\varphi,0)\\

\quad \vee \bigvee_{(r,\zeta,r') \in \Delta, \Pi \models \zeta}(\dbox{U_{r'}}\varphi, 1)\\

\quad \vee \bigvee_{(r,\zeta^-,r') \in \Delta,\Pi \models \zeta}(\dbox{U_{r'}}\varphi, -1)\\

\quad \vee \bigvee_{(r,\psi\?,r') \in \Delta} (\psi,0) \wedge (\dbox{U_{r'}}\varphi,0), & \text{if } r \in G,\\~\\

\quad \bigvee_{(r,\zeta,r') \in \Delta, \Pi \models \zeta}(\dbox{U_{r'}}\varphi, 1)\\

\quad \vee \bigvee_{(r,\zeta^-,r') \in \Delta,\Pi \models \zeta}(\dbox{U_{r'}}\varphi, -1)\\

\quad \vee \bigvee_{(r,\psi\?,r') \in \Delta} (\psi,0) \wedge (\dbox{U_{r'}}\varphi,0), & \text{otherwise}
\end{cases}$

\item $\delta(\bbox{U}\varphi,\Pi)= \begin{cases}

\quad (\varphi,0)\\

\quad \wedge \bigwedge_{r,\zeta,r' \in \Delta,\Pi \models \zeta}(\bbox{U_{r'}}\varphi,1)\\

\quad \wedge \bigwedge_{r,\zeta^-,r' \in \Delta, \Pi \models \zeta}(\bbox{U_{r'}}\varphi,-1)\\

\quad \wedge \bigwedge_{r,\psi\?,r' \in \Delta} (\nnf(\neg\psi),0) \vee (\bbox{U_{r'}}\varphi,0), & \text{if } r \in G\\~\\

\quad \bigwedge_{r,\zeta,r' \in \Delta,\Pi \models \zeta}(\bbox{U_{r'}}\varphi,1)\\

\quad \wedge \bigwedge_{r,\zeta^-,r' \in \Delta, \Pi \models \zeta}(\bbox{U_{r'}}\varphi,-1)\\

\quad \wedge \bigwedge_{r,\psi\?,r' \in \Delta} (\nnf(\psi),0) \vee (\bbox{U_{r'}}\varphi,0), & \text{otherwise.}
\end{cases}$
\end{enumerate}\vspace*{2mm}

    An empty disjunction is interpreted as \textbf{false} and an empty
    conjunction is interpreted as \textbf{true}.
    \end{definition}

Note that in 5 and 6, the two cases differ only in the existence of the first clause that depends on whether the starting state of the path automaton is an accepting state.

\begin{lemma}
  Let $\varphi$ be an ALDL$_f$ formula and let $A_{\varphi} = (Q, \Sigma, \delta, q_0, F)$ be the 2AFW that results from the construction in Definition 10. The number of states in $A_{\varphi}$, $|Q|$, is less than or equal to $|\varphi|$.

\end{lemma}

\begin{proof}
  By Lemma 1 we have that $|CL(\varphi| \leq |\varphi|$. Now, $Q = CL(\varphi)$, so $|Q| \leq |\varphi|$.
\end{proof}

\subsection{Example}

Here we provide an example of a 2AFW built from the construction in Section 6.1. The input ALDL$_f$ formula is that given as an example in Section 4.4: $\bbox{\alpha_3}\dbox{\alpha_2}\bbox{\alpha_1}\textbf{false}$. Following the construction provided in Section 6.1 results in the 2AFW $A_\varphi = (Q, \Sigma, \Delta, r, F)$, where

\begin{itemize}
\item $Q = \{$
    $\bbox{\alpha_3}\dbox{\alpha_2}\bbox{\alpha_1}\textbf{false}$,
    $\bbox{\alpha_{3_{s_1}}}\dbox{\alpha_2}\bbox{\alpha_1}\textbf{false}$,
    $\dbox{\alpha_2}\bbox{\alpha_1}\textbf{false}$,
    $\dbox{\alpha_{2_{q_1}}}\bbox{\alpha_1}\textbf{false}$,
    $\bbox{\alpha_1}\textbf{false}$,
    $\bbox{\alpha_{1_{p_1}}}\textbf{false}$,
    $a$,
    $\neg a$,
    $\textbf{false}\}$

    \item $\Sigma = \{ \{a\}, \emptyset\}$

    \item $r = \bbox{\alpha_3}\dbox{\alpha_2}\bbox{\alpha_1}\textbf{false}$

    \item $F = \{$
      $\bbox{\alpha_3}\dbox{\alpha_2}\bbox{\alpha_1}\textbf{false}$,
      $\bbox{\alpha_{3_{s_1}}}\dbox{\alpha_2}\bbox{\alpha_1}\textbf{false}$,
      $\bbox{\alpha_1}\textbf{false}$,
      $\bbox{\alpha_{1_{p_1}}}\textbf{false}$
      $\}$
\end{itemize}

and $\delta$ is defined with these transitions:

\begin{enumerate}[align=left]
\item $\delta(\bbox{\alpha_3}\dbox{\alpha_2}\bbox{\alpha_1}\textbf{false}, \{a\}) =
  (\bbox{\alpha_3}\dbox{\alpha_2}\bbox{\alpha_1}\textbf{false}, 1)  \wedge ((\neg a, 0) \vee \bbox{\alpha_{3_{s_1}}}\dbox{\alpha_2}\bbox{\alpha_1}\textbf{false}, 0)$

 \item $\delta(\bbox{\alpha_3}\dbox{\alpha_2}\bbox{\alpha_1}\textbf{false}, \emptyset) =
  (\bbox{\alpha_3}\dbox{\alpha_2}\bbox{\alpha_1}\textbf{false}, 1)  \wedge ((\neg a, 0) \vee \bbox{\alpha_{3_{s_1}}}\dbox{\alpha_2}\bbox{\alpha_1}\textbf{false}, 0)$

 \item $\delta(\bbox{\alpha_{3_{s_1}}}\dbox{\alpha_2}\bbox{\alpha_1}\textbf{false}, \{a\}) =
   (\dbox{\alpha_2}\bbox{\alpha_1}\textbf{false},0)$

 \item $\delta(\bbox{\alpha_{3_{s_1}}}\dbox{\alpha_2}\bbox{\alpha_1}\textbf{false}, \emptyset) =
   (\dbox{\alpha_2}\bbox{\alpha_1}\textbf{false},0)$

 \item $\delta(\dbox{\alpha_2}\bbox{\alpha_1}\textbf{false}, \{a\}) =
   ((\dbox{\alpha_2}\bbox{\alpha_1}\textbf{false}, -1) \vee (a, 0)) \wedge ((\dbox{\alpha_{2_{q_1}}}\bbox{\alpha_1}\textbf{false},0))$

 \item $\delta(\dbox{\alpha_2}\bbox{\alpha_1}\textbf{false}, \emptyset) =
   ((\dbox{\alpha_2}\bbox{\alpha_1}\textbf{false}, -1) \vee (a, 0)) \wedge ((\dbox{\alpha_{2_{q_1}}}\bbox{\alpha_1}\textbf{false},0))$

 \item $\delta(\dbox{\alpha_{2_{q_1}}}\bbox{\alpha_1}\textbf{false}, \{a\}) =
   (\bbox{\alpha_1}\textbf{false},0)$

 \item $\delta(\dbox{\alpha_{2_{q_1}}}\bbox{\alpha_1}\textbf{false}, \emptyset) =
   (\bbox{\alpha_1}\textbf{false},0)$

 \item $\delta(\bbox{\alpha_1}\textbf{false}, \{a\}) =
   (\bbox{\alpha_{1_{p_1}}}\textbf{false}, -1)$

 \item $\delta(\bbox{\alpha_1}\textbf{false}, \emptyset) =
   (\bbox{\alpha_{1_{p_1}}}\textbf{false}, -1)$

 \item $\delta(\bbox{\alpha_{1_{p_1}}}\textbf{false}, \{a\}) =
   (\textbf{false},0) $

 \item $\delta(\bbox{\alpha_{1_{p_1}}}\textbf{false}, \emptyset) =
   (\textbf{false},0) $

\end{enumerate}
\subsection{\texorpdfstring{Correctness of ALDL$_f$ to 2AFW Construction}{Correctness of ALDLf to 2AFW Construction}}

In this section we prove that when the above construction is applied to an ALDL$_f$
formula $\varphi$ to obtain 2AFW $A_\varphi$, $\mathcal{L}(\varphi) = \mathcal{L}(A_\varphi)$. 

In proving the equivalence of an arbitrary ALDL$_f$ formula $\varphi$ and its 2AFW $A_\varphi$ (Theorem 1 below),
it is easy to show equivalence when the ALDL$_f$ formula is propositional. It is not as straightforward for
ALDL$_f$ formulas of the form $\dbox{U}\varphi$ or $\bbox{U}\varphi$.

From the semantics of ALDL$_f$ we can see that two conditions must hold for it to be the case that
$\pi,i \models \dbox{U}\varphi$: (a) there is a $j$ such that $\pi,i,j \models U$ and (b) $\pi,j \models \varphi$.
We would like to have two corresponding properties for 2AFW runs that would allow us to show that $(\pi,i) \in \mathcal{L}(A_{\dbox{U}\varphi})$.
Condition (b) has an obvious analogue for runs of $A_{\dbox{U}\varphi}$, namely (b') that $(\pi,j) \in \mathcal{L}(A_{\varphi})$.
Condition (a) does not have an obvious analogue, so we establish a non-obvious one:
(a'): there is a run of $A_{\dbox{U}\varphi}$ on $\pi$ from $i$ that has a \textit{viable path} (Definition \ref{def:viable_path} below)
from the root to a node with configuration $(\varphi,j)$. Lemma \ref{lem:vpath_iff_vwalk} shows that (a) holds iff (a') holds. In
the proof of the main theorem, the inductive hypothesis gives us (b) and (b') when they are needed.

Intuitively, if $\pi,i,j \models U$, where $U = (R,T,\Delta,r,G)$ then there is a sequence
of transitions from $U$'s start state $r$ to some accepting state $g \in G$ that moves the
read position of $\pi$ from $i$ to $j$. We call this sequence a \textit{path automaton walk}.
We define it formally as follows:

\begin{definition}\label{def:path_automaton_walk}
    Let $U = (R,T,\Delta,r,G)$, $\pi \in (2^{AP})^+$. For all $r' \in R$ let $U_{r'} = (R,T,\Delta,r',G)$.
    A \textit{path automaton walk} of $U$ on $\pi$ from $i$ to $j$ is inductively defined as follows:
    \begin{enumerate}
        \item For all $g \in G$, for all $0 \leq j < |\pi|$, the empty sequence is a path automaton walk of $U_g$ on $\pi$ from $j$ to $j$.
        \item If there is some $(r,\zeta,r') \in \Delta$ such that for some $i$, $0 \leq i < |\pi|$, $\pi_i \models \zeta$ and there is a path automaton walk $w$
            of $U_{r'}$ on $\pi$ from $i+1$ to $j$, then $(i,\zeta,j) :: w$ is a path automaton walk of $U$ from $i$ to $j$.
        \item If there is some $(r,\zeta^-,r') \in \Delta$ such that for some $i$, $0 \leq i < |\pi|$, $\pi_i \models \zeta$ and there is a path automaton walk $w$
            of $U_{r'}$ on $\pi$ from $i-1$ to $j$, then $(i,\zeta^-,j) :: w$ is a path automaton walk of $U$ from $i$ to $j$.
        \item If there is some $(r,\psi \?,r') \in \Delta$ such that for some $i$, $0 \leq i < |\pi|$, $\pi,i \models \psi$ and there is a path automaton walk $w$
            of $U_{r'}$ on $\pi$ from $i$ to $j$, then $(i,\psi \? ,j) :: w$ is a path automaton walk of $U$ from $i$ to $j$.
    \end{enumerate}
    A path automaton walk $w$ is a \textit{viable path automaton walk} if, for all elements of the form $(i, \psi \? ,j)$
    occurring in $w$, there is an accepting run of $A_\psi$ on $\pi$ from $i$.

\end{definition}

\begin{lemma}\label{lem:walk_iff_relation}
    Let $U = (R,T,\Delta,r,G)$, $\pi \in (2^{AP})^+$, $0 \leq i < |\pi|$, $0 \leq j < |\pi|$. Then $\pi,i,j \models U$ if and only if
  there is a viable path automaton walk of $U$ on $\pi$ from $i$ to $j$.
\end{lemma}

We introduce the concept of a \textit{viable path} from the root $t$ of a run $\rho$
to some descendant $v$ of $\rho$ to describe a property of 2AFW runs that is the analogue
of the relation $\pi,i,j \models U$ for ALDL$_f$ formulas. Formally, we have the following:

\begin{definition}\label{def:viable_path}
  Let $\varphi$ be an ALDL$_f$ formula, let $A_\varphi$ be the 2AFW that is the result of applying
  the above construction to $\varphi$, and let $\pi \in (2^{AP(\varphi)})^*$. Let $\rho$ be a run of
  $A_\varphi$ on $\pi$, and let $t$ be the root of $\rho$. Let $h$
  be a path from $t$ to $v$. Then $h$ is a $\textit{viable path}$ if, for all nodes $u$ in $\rho$,
  if $u$ shares a parent with some node $v$ of $\rho$, $u \neq v$, that is the start or destination of an edge in $h$ then $u$ is the
  root of an accepting subtree.
\end{definition}

The following lemma formalizes the notion that the viable path is to runs what the relation
$\pi,i,j \models U$ is to the semantics of ALDL$_f$:

\begin{lemma}\label{lem:vpath_iff_vwalk}
    Let $\varphi = \dbox{U}\varphi'$, where $U = (R,T,\Delta,r,G)$ and $\varphi'$ is
    an ALDL$_f$ formula. Let $\pi \in (2^{AP(\varphi)})^+$, $0 \leq i < |\pi|$, $0 \leq j < |\pi|$.
    Then there is some run $\rho$ of $A_{\varphi}$ on $\pi$ from $i$ such that there is a viable path
    from the root $t$ of $\rho$ to some node $v$ such that $\config(v) = (\varphi',j)$ if and only if
    there is a viable path automaton walk of $U$ on $\pi$ from $i$ to $j$.
\end{lemma}

We are now ready to formally state the main theorem:

\begin{theorem}\label{thm:aldlf_2afw_equiv}
  Let $\varphi$
  be an ALDL$_f$ formula in negation normal form and let $A_\varphi$ be the two-way alternating
  automaton resulting from the previous construction. Let $\pi \in (2^{AP})^+$. Then $\pi,i
  \models \varphi$ if and only if there is an accepting run of $A_\varphi$ on $\pi$ from $i$. 
\end{theorem}

\begin{proof} (Theorem \ref{thm:aldlf_2afw_equiv})

    ($\rightarrow$):

    We assume $\pi,i \models \varphi$ and show that there is an accepting run of $A_\varphi$ on $\pi$ from $i$.
    We proceed by structural induction on $\varphi$.

    \textbf{Base case:}

    \textit{Case (1):} $\varphi = P$, where $P \in AP$.

    From $\pi,i \models P$ and rule 1 of Definition \ref{def:aldlf_semantics} we have that $P \in \pi_i$.  Let $\rho$ be
    a run of $A_P$ on $\pi$ from $i$. The root $t$ of $\rho$ has configuration $(P,i)$. By rule 1 of Definition \ref{def:aldlf_to_2afw}
    we have that $\delta(P,\pi_i) = \textbf{true}$, so $t$ has one child $c$, where $\config(c) = Accept$. Now, $\rho$ is an accepting
    run because all of its branches are accepting.

    \textit{Case (2):} $\varphi = \neg P$, where $P \in AP$.

    From $\pi,i \models \neg P$ and rule 2 of Definition \ref{def:aldlf_semantics} we have that $\pi,i \not\models P$. By rule
    1 of Definition \ref{def:aldlf_semantics} we have that $P \not\in \pi_i$. Let $\rho$ be a run of $A_P$ on $\pi$ from $i$ and let
    $t$ be the root of $\rho$. Now, $\config(t) = (\neg P,i)$. By rule 2 of Definition \ref{def:aldlf_to_2afw}
    we have that $\delta(\neg P,\pi_i) = \textbf{true}$, so $t$ has one child $c$, where $\config(c) = Accept$. Now, $\rho$ is an
    accepting run because all of its branches are accepting.

    \textbf{Inductive step:}

    \textit{Case (1):} $\varphi = \varphi_1 \wedge \varphi_2$.

    From rule 3 of Definition \ref{def:aldlf_semantics} we have that
    $\pi,i \models \varphi_1$ and $\pi,i \models \varphi_2$.
    By the inductive hypothesis we have that there are accepting runs $\rho_{\varphi_1}$ and $\rho_{\varphi_2}$ of $A_{\varphi_1}$
    on $\pi$ from $i$
    and $A_{\varphi_2}$ on $\pi$ from $i$. Let $t_1$ and $t_2$ be the roots of $\rho_{\varphi_1}$ and $\rho_{\varphi_2}$, respectively.
    Note that $\config(t_1) = (\varphi_1,i)$ and $\config(t_2) = (\varphi_2,i)$.

    Let $\theta = \delta(\varphi_1 \wedge \varphi_2,\pi_i)$. By rule 3 of Definition \ref{def:aldlf_to_2afw},
    $\theta = (\varphi_1,0) \wedge (\varphi_2,0)$. So we have that
    $\{(\varphi_1,0), (\varphi_2,0)\} \in \mymodels(\theta)$.
    Moreover, we have that $(\varphi_1 \wedge \varphi_2,i) \sderives{\varphi_1,0} (\varphi_1,i)$
    and $(\varphi_1 \wedge \varphi_2,i) \sderives{\varphi_2,0} (\varphi_2,i)$.

    Let $\rho$ be a configuration-labeled tree with root $t$, where $\config(t) = (\varphi_1 \wedge \varphi_2,i)$
    and $t$ has two children, $t_1$ and $t_2$. Now, $\rho$ is a run of $A_{\varphi_1 \wedge \varphi_2}$ on $\pi$ from $i$.
    Moreover, because $\rho_{\varphi_1}$ and $\rho_{\varphi_2}$ are accepting runs it follows that $\rho$ is an accepting run.

    \textit{Case (2):} $\varphi = \varphi_1 \vee \varphi_2$.

    From rule 4 of Definition \ref{def:aldlf_semantics} we have that at least one of the following is true: (a) $\pi,i \models \varphi_1$
    or (b) $\pi,i \models \varphi_2$. Without loss of generality we assume (a). By the inductive hypothesis we have that there is
    an accepting run $\rho_{\varphi_1}$ of $A_{\varphi_1}$ on $\pi$ from $i$. Let $t_1$ be the root of $\rho_{\varphi_1}$. Note that
    $\config(t_1) = (\varphi_1,i)$.

    Let $\theta = \delta(\varphi_1 \vee \varphi_2,\pi_i)$. By rule 4 of Definition \ref{def:aldlf_to_2afw}
    $\theta = (\varphi_1,0) \vee (\varphi_2,0)$.
    So we have that $\{(\varphi_1,0 \} \in \mymodels(\theta)$. Moreover, we have that
    $(\varphi_1 \vee \varphi_2,i) \sderives{\varphi_1,0} (\varphi_1,i)$.

    Let $\rho$ be a configuration-labeled tree with root $t$, where $\config(t) = (\varphi_1 \vee \varphi_2,i)$
    and $t$ has one child, $t_1$. Now, $\rho$ is a run of $A_{\varphi_1 \vee \varphi_2}$ on $\pi$ from $i$.
    Moreover, because $\rho_{\varphi_1}$ is an accepting run it follows that $\rho$ is an accepting run.

    \textit{Case (3):} $\varphi = \dbox{U}\varphi'$, where $U = (R,T,\Delta,r,G)$.

    From $\pi,i \models \dbox{U}\varphi'$ we have that there is some $j$ such that $0 \leq j < |\pi|$,
    $\pi,i,j \models U$, and $\pi,j \models \varphi'$. By the inductive hypothesis we have that there is an accepting run
    $\rho_{\varphi'}$ of $A_{\varphi'}$ on $\pi$ from $j$.
    By Lemma \ref{lem:walk_iff_relation} we have that there is a path automaton walk $w$ of $U$ on
    $\pi$ from $i$ to $j$. From Definition \ref{def:path_automaton_walk}, it is the case that for every $(k, \psi \? ,l)$ in $w$ we have that $\pi,k \models \psi$.
    From the inductive hypothesis we have that for every $(k, \psi \? ,l)$ in $w$ there is an accepting run $\rho_{\psi,k}$ of $A_\psi$ on
    $\pi$ from $k$. So by Definition \ref{def:path_automaton_walk}, $w$ is a viable path automaton walk of $U$ on $\pi$ from $i$ to $j$.
    By Lemma \ref{lem:vpath_iff_vwalk}
    there is some run $\rho$ of $A_{\dbox{U}\varphi'}$ on $\pi$ from $i$ with a viable path from the root $t$ of $\rho$ to a node
    $v$ with $\config(\varphi',j)$. Let $\rho'$ be identical to $\rho$ except that the subtree with $v$ as root is replaced by
    $\rho_{\varphi'}$. Now, $\rho_{\varphi'}$ is an accepting subtree, and since there is a viable path from $t$ to $\rho_{\varphi'}$
    it is the case that every other subtree in $\rho'$ is accepting. So $\rho'$ is an accepting run of $A_{\dbox{U}\varphi'}$ on $\pi$ from $i$.

    \textit{Case (4):} $\varphi = \bbox{U}\varphi'$, where $U = (R,T,\Delta,r,G)$.

    The reasoning from Case (3) applies here, mutatis mutandis.

    ($\leftarrow$):

    We assume there is an accepting run $\rho$ of $A_\varphi$ on $\pi$ from $i$ and show that $\pi,i \models \varphi$.
    We proceed by structural induction on $\varphi$.

    \textbf{Base case:}

    \textit{Case (1):} $\varphi = P$, where $P \in AP$.

    Let $t$ be the root of $\rho$. Now, $\config(t) = (P, i)$. Let $\theta = \delta(P,\pi_i)$. From rule 1 of
    Definition \ref{def:aldlf_to_2afw} we have that either (a) $\theta = \textbf{true}$ or (b) $\theta = \textbf{false}$.
    But if (b) holds then $t$ has no children and $\rho$ is a rejecting run\----a contradiction. So (a) must hold, so
    $\pi_i \models P$, so $\pi,i \models P$.

    \textit{Case (2):} $\varphi = \neg P$, where $P \in AP$.

    Let $t$ be the root of $\rho$. Now, $\config(t) = (\neg P,i)$. Let $\theta = \delta(\neg P,\pi_i)$.
    From rule 2 of Definition \ref{def:aldlf_to_2afw} we have that either (a) $\theta = \textbf{true}$ or
    (b) $\theta = \textbf{false}$. But if (b) holds then $t$ has no children, so $\rho$ is a rejecting run\----a contradiction.
    So (a) must hold, so $\pi_i \not\models P$, so $\pi,i \not\models P$, so $\pi,i \models \neg P$.

    \textbf{Inductive step:}

    \textit{Case (1):} $\varphi = \varphi_1 \wedge \varphi_2$.

    Let $t$ be the root of $\rho$. Now, $\config(t) = (\varphi_1 \wedge \varphi_2,i)$.
    Let $\theta = \delta(\varphi_1 \wedge \varphi_2,\pi_i) = (\varphi_1,0) \wedge (\varphi_2,0)$.
    It is easy to see that the only member of $\mymodels(\theta)$ is $\{(\varphi_1,0),(\varphi_2,0)\}$.
    Moreover, $(\varphi_1 \wedge \varphi_2,i) \sderives{(\varphi_1,0)} (\varphi_1,i)$ and
    $(\varphi_1 \wedge \varphi_2,i) \sderives{(\varphi_2,0)} (\varphi_2,i)$. So $t$ must have children
    $c_1$ and $c_2$ such that $\config(c_1) = (\varphi_1,i)$ and $\config(c_2) = (\varphi_2,i)$. Let
    $\rho_{\varphi_1}$ be the subtree with $c_1$ as root and $\rho_{\varphi_2}$ be the subtree with $c_2$
    as root. Now, $\rho_{\varphi_1}$ is a run of $A_{\varphi_1}$ on $\pi$ from $i$ and $\rho_{\varphi_2}$ is
    a run of $A_{\varphi_2}$ on $\pi$ from $i$. Because $\rho$ is an accepting run, all of its subtrees
    must be accepting runs. In particular, $\rho_1$ and $\rho_2$ are accepting runs. By the inductive hypothesis
    we have that $\pi,i \models \varphi_1$ and $\pi,i \models \varphi_2$. But by rule 3 of Definition \ref{def:aldlf_semantics}
    we have that $\pi,i \models \varphi_1 \wedge \varphi_2$.

    \textit{Case (2):} $\varphi = \varphi_1 \vee \varphi_2$.

    Let $t$ be the root of $\rho$. Now, $\config(t) = (\varphi_1 \vee \varphi_2,i)$.
    Let $\theta = \delta(\varphi_1 \vee \varphi_2,\pi_i) = (\varphi_1,0) \vee (\varphi_2,0)$. It is easy to
    see that $\mymodels(\theta) = \{ \{ (\varphi_1,0),(\varphi_2,0)  \}, \{ (\varphi_1,0)  \}, \{ (\varphi_2,0) \} \}$.
    Moreover $(\varphi_1 \vee \varphi_2,i) \sderives{(\varphi_1,0)} (\varphi_1,i)$
    and $(\varphi_1 \vee \varphi_2,i) \sderives{(\varphi_2,0)} (\varphi_2,i)$.

    So either (a) $t$ has exactly two children, $c_1$ and $c_2$, where $\config(c_1) = (\varphi_1,i)$ and
    $\config(c_2) = (\varphi_2,i)$, (b) $t$ has exactly one child, $c_1$, where $\config(c_1) = (\varphi_1,i)$,
    or (c) $t$ has exactly one child, $c_2$, where $\config(c_2) = (\varphi_2,i)$.

    Suppose (a). Let $\rho_{\varphi_1}$ be the subtree with $c_1$ as root and $\rho_{\varphi_2}$ be the subtree
    with $c_2$ as root. Now, $\rho_{\varphi_1}$ is a run of $A_{\varphi_1}$ on $\pi$ from $i$ and $\rho_{\varphi_2}$
    is a run of $A_{\varphi_2}$ on $\pi$ from $i$. Because $\rho$ is an accepting run, all of its subtrees must
    be accepting runs. In particular, $\rho_{\varphi_1}$ and $\rho_{\varphi_2}$ are accepting runs. By the inductive hypothesis
    we have that $\pi,i \models \varphi_1$ and $\pi,i \models \varphi_2$. But by rule 4 of Definition \ref{def:aldlf_semantics}
    we have that $\pi,i \models \varphi_1 \vee \varphi_2$.

    Suppose (b). Let $\rho_{\varphi_1}$ be the subtree with $c_1$ as root. Now, $\rho_{\varphi_1}$ is a run of $A_{\varphi_1}$
    on $\pi$ from $i$. Because $\rho$ is an accepting run, all of its subtrees must be accepting runs. In particular,
    $\rho_{\varphi_1}$ is an accepting run. By the inductive hypothesis we have that $\pi,i \models \varphi_1$. But by rule 4
    of Definition \ref{def:aldlf_semantics} we have that $\pi,i \models \varphi_1 \vee \varphi_2$.

    Suppose (c). Let $\rho_{\varphi_2}$ be the subtree with $c_2$ as root. Now, $\rho_{\varphi_2}$ is a run of $A_{\varphi_2}$ on
    $\pi$ from $i$. Because $\rho$ is an accepting run, all of its subtrees must be accepting runs. In particular, $\rho_{\varphi_2}$
    is an accepting run. By the inductive hypothesis we have that $\pi,i \models \varphi_2$. But by rule 4 of
    Definition \ref{def:aldlf_semantics} we have that $\pi,i \models \varphi_1 \vee \varphi_2$.

    \textit{Case (3):} $\varphi = \dbox{U}\varphi'$, where $U = (R,T,\Delta,r,G)$.

    From rule 5 of Definition \ref{def:aldlf_to_2afw} we can see that every node with configuration of the form $(\dbox{U}\varphi',n)$
    must either (1) have a descendant with configuration of the form $(\dbox{U_{r'}}\varphi',n')$, where $r' \in R$, (2) have a descendant with
    configuration of the form $(\varphi',n')$, or (3) have no descendants (this can occur if $r \not\in G$ and $U$ cannot make
    any transitions, rendering all of the disjunctions empty). It cannot be the case that (3) holds because if $\rho$ has a finite
    branch with a leaf that does not have configuration $Accept$ then $\rho$ is a rejecting run---a contradiction. Let $t$ be the root
    of $\rho$. If (1) holds for every descendant of $t$ with configuration of the form $(\dbox{U_{r'}}\varphi',n')$ then
    $\rho$ has an infinite branch that has no nodes with state-label of the form $\bbox{\alpha}\psi$, so $\rho$ is a rejecting run---a contradiction.
    So (2) must hold for some descendant of $t$. So there is a path $h$ from $t$ to some node $v$,
    where $\config(v) = (\varphi',j)$ and $0 \leq j < |\pi|$.

    Since $\rho$ is an accepting run all of its branches are accepting, so $h$ is a viable path from $t$ to $v$.
    By Lemma \ref{lem:vpath_iff_vwalk} there is a viable path automaton walk $w$ of $U$ on $\pi$ from $i$ to $j$. But since $w$ is a viable
    path automaton walk, $w$ is a path automaton walk, so by Lemma \ref{lem:walk_iff_relation} we have that $\pi,i,j \models U$.
    Moreover, since $\rho$ contains an accepting subtree with configuration $(\varphi',j)$ we have that $\pi,j \models \varphi'$.
    But since $\pi,i,j \models U$ and $\pi,j \models \varphi'$, we have that $\pi,i \models \dbox{U}\varphi'$.

    \textit{Case (4):} $\varphi = \bbox{U}\varphi'$, where $U = (R,T,\Delta,r,G)$.

    The reasoning from Case (3) applies here, mutatis mutandis.
\end{proof}

\section{From 2AFWs to NFWs}
\label{ch:2afwtonfw}

Now we proceed to the final step of the ALDL$_f$ to NFA translation: the construction of an NFA from a 2AFW. The technique we use here is a generalization of that given by Chandra, Kozen, and Stockmeyer to construct an NFA from a one-way alternating finite automaton \cite{chandra}.

As one step in this process, in Theorem 2 we model the acceptance of a word $w$ by a 2AFW $A$ as a two-player game in which a Protagonist tries to show that $A$ accepts $w$ and an Antagonist tries to show that it does not. Before starting on Theorem 2 it is necessary to introduce some basic game-related definitions regarding strategies. In general, a strategy is a function that takes as input the state of a game and returns the next move that should be made. In Section 7.1 we define the \textit{history strategy} for 2AFWs, which requires as input the entire history of a game in order to prescribe the next move. In Section 7.2 we show that the full history is not needed and one only needs to use the current configuration of the game to prescribe the next move and so we define the \textit{strategy word} for 2AFWs, which only takes into account the current configuration of the game in order to prescribe moves.

\subsection{Definitions and Lemmas}

\subsubsection{History Strategy}
Let $A = (Q,\Sigma,\delta,q_0,F)$ be a 2AFW. A \textit{history} for $A$ is an element of
$(Q \times \mathbb{N})^*$. Intuitively, a history is a sequence of configurations of $A$ that
represent a (possibly partial) branch of a run $\rho$ of $A$ on some word $w$.
A \textit{history strategy} is a mapping
$h : (Q \times \mathbb{N})^* \rightarrow 2^{\{-1,0,1 \} \times Q  }$.
Intuitively, a history strategy takes as input the current history of a game and prescribes the next move.

A history strategy $h$ is \textit{on} an input word $w$ if the following conditions apply:

\begin{enumerate}
    \item $h(\epsilon) = (q_0,0)$
    \item If $p$ is a history and the last element of $p$ has a configuration with index 0, then $h(p)$ has no elements of the form $(-1,q')$
    \item If $p$ is a history and the last element of $p$ has a configuration with index $|w| - 1$, then $h(p)$ has no elements of the form $(1,q')$
    \item If $p$ is a history and the last element of $p$ has a configuration of the form $(q,u)$, then $\delta(q,w_u)$ is satisfied by $h(p)$, where $w_u$ is the element of $w$ at index $u$.
\end{enumerate}

The first condition ensures that the strategy has a prescription for the start of the game, before any moves have been made.
The next two conditions ensure that the strategy does not prescribe a move that would lead to an invalid index,
while the final condition ensures that each transition satisfies $A$'s transition function.

A \textit{history path} $\beta$ in a history strategy $h$ on input $w$ is a maximal length sequence of pairs from $|w| \times Q$
that may be finite or infinite and has the following property: if $\beta = (u_1,q_1), (u_2,q_2), \ldots$ is infinite, then,
for all $i \geq 0$, there is some $d_{i+1} \in \{-1,0,1 \}$ such that $(d_{i+1},q_{i+1}) \in h((u_1,q_1),\ldots,(u_i,q_i)$
(if $\beta$ is finite, this condition holds for all $0 \leq i < |\beta|$).

Thus, $\beta$ is obtained by following transitions in the history strategy.
We define $\inf(\beta)$ to be the set of states in $Q$ that occur infinitely often in $\beta$.
We say that an infinite history path $\beta$ is accepting if there is a $q \in F$ such that $q \in \inf(\beta)$.
We say that a finite history path $\beta$ of length $n$ is accepting if $\delta(q_n, w_{u_n}) = \textbf{true}$. We say that a history strategy
$h$ is accepting if all of its history paths are accepting.

\subsubsection{Strategy Word}

Let $A = (Q,\Sigma,\delta,q_0,F)$ be a 2AFW and let $w$ be a word over $\Sigma$. A \textit{strategy word} for $A$ on $w$ is a mapping $\gamma : |w| \rightarrow 2^{Q \times \{-1,0,1\} \times Q}$.
For each label $l \subseteq Q \times \{-1,0,1\} \times Q$ we define $\origin(l) = \{ q \mid (q,d,q') \in l \}$. Moreover, $q_0 \in \origin(\gamma(0))$ and for each position $u$, $0 \leq u < |w|$, and each state $q \in \origin(\gamma(u))$ the set $\{ (d,q') \mid (q,d,q') \in \gamma(u) \}$ satisfies $\delta(q, w_u)$. Also, $\gamma(0)$ must not contain elements of the form $(q,-1,q')$ and $\gamma(|w|-1)$ must not contain elements of the form $(q,1,q')$; these conditions ensure that the strategy word does not attempt to move the read position of the input to a location before the beginning of $w$ or after the end of $w$.

Intuitively, a label defines a set of transitions; $q$ and $q'$ prescribe origin and destination states, while $d$ prescribes the direction to move the read head of the input word. Each label can be viewed as a strategy of satisfying the transition function. A strategy word is similar to a history strategy, but it makes prescriptions based only on the current configuration of the game, not the history of the game. Note that because the alphabet can express a transition occurring in one of three directions between any pair of states, the size of the alphabet of strategy words is $8^{|Q|^2}$.

A \textit{strategy path} $\beta$ in a strategy word $\gamma$ on an input word $w$ is a maximal sequence of pairs
from $|w| \times Q$ that may be infinite
or finite and has the following property: if $\beta = (u_1,q_1), (u_2,q_2),\ldots$ is infinite then, for all $i \geq 0$,
there is some $d_i \in \{-1,0,1\}$ such that $(q_i,d_i,q_{i+1}) \in \gamma(u_i)$ and $u_{i+1} = u_i + d_i$.
Thus, $\beta$ is obtained by following transitions in the strategy word. We define
$\inf(\beta)$ to be the set of states in $Q$ that occur infinitely often in $\beta$. We say that an infinite strategy path $\beta$ is
accepting if there is a $q \in F$ such that $q \in \inf(\beta)$.
If $\beta = (u_1,q_1),(u_2,q_2),\ldots,(u_n,q_n)$ is finite then, for all $0 \leq i < n$, there is
some $d_i \in \{-1,0,1\}$ such that $(q_i,d_i,q_{i+1}) \in \gamma(u_i)$ and $u_{i+1} = u_i + d_i$.
Moreover, $\delta(q_n,w_{u_n}) = \textbf{true}$; this condition ensures that strategy paths are maximal, i.e., they do not end while it is still possible
to make transitions. We say that a finite strategy path $\beta$
is accepting if $\delta(q_n,w_{u_n}) = \textbf{true}$. We say that the strategy word $\gamma$ on $w$ is accepting if all strategy paths in $\gamma$ are accepting.

\begin{lemma}
Let $\gamma$ be a strategy word on an input word $w$. Then if $\beta$ is a finite strategy path in $\gamma$, then $\beta$ is accepting.
\end{lemma}
\begin{proof}
Because $\beta$ is a finite strategy path, it is of the form $(u_1,q_1),(u_2,q_2),\ldots,(u_n,q_n)$. From the definition of strategy path we have that $(q_i,d_i,q_i+1) \in \gamma(u_i)$ and $u_i+1 = u_i + d_i$ for all $0 < i \leq n$. Because $\gamma$ is on $w$, we have that, for all $0 < i \leq n$, $(d_i,q_{i+1})$ is an element of some set that satisfies $\delta(q_i,w_u)$. So $\delta(q_n,w_{u_n})$ cannot be \textbf{false}, because \textbf{false} has no satisfying set. Nor can $\delta(q_n,w_{u_n})$ be an arbitrary formula; it must be \textbf{true}, else it would not be the final element of $\beta$. So $\beta$ is an accepting strategy path.
\end{proof}

\subsubsection{Annotation}
Let $A = (Q,\Sigma,\delta,q_0,F)$ be a 2AFW, $w \in \Sigma^*$, and $\gamma : |w| \rightarrow 2^{Q \times \{-1,0,1\} \times Q}$ be a
strategy word for $A$ on $w$. For each state $q \in Q$, let $\final(q) = 1$ if $q \in F$, $0$ otherwise.

Informally, an annotation keeps track of relevant information regarding all finite strategy paths in $\gamma$. In particular,
for all $0 \leq u < |w|$, we want to know all the pairs of states $q,q'$ such that when $A$ is in state $q$ at
position $u$ it is possible for $A$ to engage in a series of transitions that leave it in state $q'$
after having returned to position $u$. Moreover, we want to know whether this series of transitions included a
step that involved transitioning to an accepting state.

Formally, an \textit{annotation}
of $\gamma$ is a mapping $\eta : |w| \rightarrow 2^{Q \times \{ 0,1 \} \times Q}$ that satisfies the
following closure conditions:

\begin{enumerate}
    \item If $(q,0,q') \in \gamma(u)$, then $(q,\final(q'),q') \in \eta(u)$

    \item If $u_1 = u_2 - 1$, $(q,1,q') \in \gamma(u_1)$, and $(q',-1,q'') \in \gamma(u_2)$, then $(q,\max(\final(q'),\final(q'')),\allowbreak q'') \in \eta(u_1)$

    \item If $u_1 = u_2 + 1$, $(q,-1,q') \in \gamma(u_1)$, and $(q',1,q'') \in \gamma(u_2)$, then $(q,\max(\final(q'),\final(q'')),\allowbreak q'') \in \eta(u_1)$

    \item If $(q, f_1, q') \in \eta(u)$ and $(q', f_2, q'') \in \eta(u)$, then $(q, \max(f_1,f_2),q'') \in \eta(u)$
\end{enumerate}
An annotation is \textit{accepting} if it has no elements of the form $(q,0,q)$. Note that because an annotation is a mapping to sets of transitions between two states that have potentially two directions, the alphabet of annotations is of size $4^{|Q|^2}$.

It is useful to know where in the input cycles can occur; we introduce a special word called a \emph{semi-path} to keep track of this. Informally, a semi-path adds positional information to the state-to-state cycle information from an annotation to indicate at which positions in the word cycles may occur. Correctness of the indexes is ensured by following transitions in the strategy word in a fashion similar to a strategy path. Formally, a \textit{semi-path} on a strategy word $\gamma$ and an annotation $\eta$ is a finite sequence $c_1,\ldots,c_m$, where for all $1 \leq i < m$ either $c_i \in |w| \times Q$ or $c_i \in |w| \times Q \times \{0,1\} \times Q$ such that the following conditions hold:

\begin{enumerate}
    \item If $c_i = (j,q_1,f,q_2)$ then $(q_1,f,q_2) \in \eta(j)$.

    \item If $c_i = (j,q)$ and $c_{i+1} = (j',q')$ then there is some $d \in \{-1,0,1\}$ such that $(q,d,q') \in \gamma(j)$ and $j' = j + d$.

    \item If $c_i = (j,q)$ and $c_{i+1} = (j',q_1,f,q_2)$ then there is some $d \in \{-1,0,1\}$ such that $(q,d,q_1) \in \gamma(j)$ and $j' = j+d$.

    \item If $c_i = (j,q_1,f,q_2)$ and $c_{i+1} = (j',q)$ then there is some $d \in \{-1,0,1\}$ such that $(q_2,d,q) \in \gamma(j)$ and $j' = j + d$.

    \item If $c_i = (j,q_1,f,q_2)$ and $c_{i+1} = (j',q_1',f',q_2')$ then there is some $d \in \{-1,0,1\}$ such that $(q_2,d,q_1') \in \gamma(j)$ and $j' = j + d$.

\end{enumerate}
For an element $c_i$ of $p$, we say that $\myindex(c_i)=j$ if $c_i=(j,q)$ or $c_i=(j,q,f,q')$. An element $c_i=(j,q,f,q')$ is called a \emph{cycle}.

Let $p$ be a semi-path of length $n$. An element $c_i$ of $p$ is \emph{accepting}, if either $c_i=(j,q,f,q')$ with $f=1$, or $c_i=(j,q)$ with $i>1$ and $q\in F$.
We say that $p$ is \emph{accepting} if it has an accepting element, that is, an element of the form $(j',q')$ with $q' \in F$ or an element of the form $(j',q_1,f,q_2)$ with $f = 1$. Finally, we say that $p$ is a $(j,q)$-\textit{semi-path} if either (1) $n>1$ and $p_0$ is either $(j,q)$ or $(j,q,f,q')$ and $p_{n-1}$ is either $(j,q)$ or $(j,q'',f,q)$, or (2) $n=1$ and $p_0 = (j,q,f,q)$. Note that a non-accepting $(j,q)$-semi-path of length 1 must be of the form $(j,q,f,q)$ with $f=0$.

\begin{lemma} (Shortening Lemma)\\
Let $A = (Q,\Sigma,\delta,q_0,F)$ be a 2AFW, $w \in \Sigma^+$, $\gamma$ be a strategy word for $A$ on $w$, and $\eta$ be an annotation of $\gamma$. Let $p$ be a  non-accepting $(j,q)$-semi-path on $\gamma$ and $\eta$ of length $n>1$. Then there is a non-accepting $(j,q)$-semi-path $c'$ on $\gamma$ and $\eta$ of length $n-1$.
\end{lemma}

\begin{proof}
Suppose first that $p$ contains two adjacent elements $c_i=(j,q_1)$ and $c_{i+1}=(j,q_2)$, then we can combine them and replace them by a new element $c=(j,q_1,f,q_2)$, with $f=\final(q_2) = 0$ due to conditions (1) and (2) of the definition of semi-path and because for all $q$ that occur in $p$ we have that $\final(q) = 0$ because $p$ is non-accepting.  Note that by the closure properties of annotations we have that $(q_1,0,q_2)\in\eta(j)$, so $p'=c_1,\ldots,c_{i-1},c,c_{i+2},\ldots,c_{|p|}$ is also a non-accepting semi-path. Similarly, if $c_i=(j,q_1,0,q_2)$ and $c_{i+1}=(j,q_2,0,q_3)$. Then we can combine them and replace them by a new element $(j,q_1,0,q_3)$ due to conditions (1) and (5). Similarly, if $c_i = (j,q_1)$ and $c_{i+1} = (j,q_2,0,q_3)$, then  we can combine them into a new element $c=(j,q_1,0,q_3)$ due to conditions (1) and (3). Also, if $c_i = (j,q_1,0,q_2)$ and $c_{i+1} = (j,q_3)$,  then  we can combine them into a new element $c=(j,q_1,0,q_3)$ due to conditions (1) and (4). Thus, we can assume that $p$ does not have two adjacent elements with the same index.

It follows that there is an element $c_i$ such that $\myindex(c_i)$ is maximal and $i$ is maximal. Then there are the following eight cases to consider, depending on whether $c_{i-1}$, $c_i$, and $c_{i+1}$ are cycles.
\begin{enumerate}
\item
$c_{i-1}=(j-1,q_1)$, $c_i=(j,q_2)$, and $c_{i+1}=(j-1,q_3)$: In this case we can combine $c_{i-1}$, $c_i$, and $c_{i+1}$ into a new cycle $(j-1,q_1,0,q_3)$.

\item
$c_{i-1}=(j-1,q_1,0,q_2)$, $c_i=(j,q_3)$, and $c_{i+1}=(j-1,q_4)$: In this case we can combine $c_{i-1}$, $c_i$, and $c_{i+1}$ into a new cycle $(j-1,q_1,0,q_4)$.

\item
$c_{i-1}=(j-1,q_1,0,q_2)$, $c_i=(j,q_3,0,q_4)$, and $c_{i+1}=(j-1,q_5)$: In this case we can combine $c_{i-1}$, $c_i$, and $c_{i+1}$ into a new cycle $(j-1,q_1,0,q_4)$.

\item
$c_{i-1}=(j-1,q_1,0,q_2)$, $c_i=(j,q_3,0,q_4)$, and $c_{i+1}=(j-1,q_5,0,q_6)$: In this case we can combine $c_{i-1}$, $c_i$, and $c_{i+1}$ into a new cycle $(j-1,q_1,0,q_6)$.
\item
$c_{i-1}=(j-1,q_1)$, $c_i=(j,q_2)$, $c_{i+1} = (j-1,q_3,f,q_4)$: In this case we can combine $c_{i-1}$, $c_i$, and $c_{i+1}$ into a new cycle $(j-1,q_1,0,q_4)$.

\item
$c_{i-1}=(j-1,q_1)$, $c_i=(j,q_2,0,q_3)$, $c_{i+1}=(j-1,q_4)$:
In this case we can combine $c_{i-1}$, $c_i$, and $c_{i+1}$ into a new cycle $(j-1,q_1,0,q_4)$.

\item
$c_{i-1}=(j-1,q_1)$, $c_i=(j,q_2,0,q_3)$, $c_{i+1}=(j-1,q_4,0,q_5)$:
In this case we can combine $c_{i-1}$, $c_i$, and $c_{i+1}$ into a new cycle $(j-1,q_1,0,q_5)$.

\item
$c_{i-1}=(j-1,q_1,0,q_2)$, $c_i=(j,q_3)$, $c_{i+1}=(j-1,q_4,0,q_5)$:
In this case we can combine $c_{i-1}$, $c_i$, and $c_{i+1}$ into a new cycle $(j-1,q_1,0,q_5)$. \qedhere

 \end{enumerate}
\end{proof}

\section{Main Proofs}

Now that we have provided the necessary definitions and lemmas, we are ready to proceed to the main proofs of this section.

\begin{theorem}
    A two-way alternating automaton on finite words $A = (Q, \Sigma, \delta, q_0, F)$
    accepts an input word $w$ iff $A$ has an accepting strategy word $\gamma$ over $w$.
\end{theorem}

\begin{proof}

    Let $G_{A,w}$ be the following game between two players, the Protagonist and Antagonist.
    Intuitively, the Protagonist is trying to show that $A$ accepts the input word $w$, and the Antagonist is trying
    to show that it does not. A $\textit{configuration}$ of the game is a pair in $Q \times \mathbb{N}$. The initial configuration is
    $(q_0, 0)$. Consider a configuration $(q,n)$. If $\delta(q,w_n) = \textbf{true}$ then the Protagonist wins immediately, and if
    $\delta(q,w_n) = \textbf{false}$ then the Antagonist wins immediately.
    Otherwise, the Protagonist chooses a set $\{(q_1,c_1), \ldots, (q_m,c_m)\}$ that
    satisfies $\delta(q, w_n)$. The Antagonist responds by choosing an element $(q_i, c_i)$ from the set. The new configuration is then
    $(q_i, n + c_i)$. If $n + c_i$ is undefined then the Antagonist wins immediately.
    Consider now an infinite play $\iota$. Let $\inf(\iota)$ be the set of states in $Q$ that repeat infinitely
    in the sequence of configurations in $\iota$. The Protagonist wins if $\inf(\iota) \cap F \neq \emptyset$.

    Suppose $A$ accepts $w$. Then there is an accepting run $\rho$ of $A$ on $w$. From $\rho$ we can obtain a winning history
    strategy $h$ for the Protagonist in the game $G_{A,w}$; we will use the existence of $h$ to prove that the Protagonist
    has a winning strategy word $\gamma$.
    Let $\mypath$ be a function that takes as input a node $x$ in $\rho$ and returns a sequence
    $s \in (Q \times \mathbb{N})^*$ such that $s$ is the sequence of configuration labels encountered when traversing from the
    root to $x$ (if $x$ is the root, then $\mypath$ returns $\epsilon$). Let $f$ be a function that takes as input
    a node $x$ in $\rho$, which has configuration $(q,u)$, and returns the set $\{(d,q') \mid (q, u+d) \text{ is the
    configuration of some child of } x \}$.

    Let $h$ be a history strategy defined as follows: $h(\epsilon) = (q_0,0)$ and for every node $x$ in $\rho$, $h(\mypath(x)) = f(x)$. Intuitively, following $h$ will result in a run identical to $\rho$.

    Now, $h$ is on $w$. Suppose not, that is, suppose that $h$ is not on $w$. Then one of the following must hold:
    \begin{enumerate}
        \item $h(\epsilon) \neq (q_0,0)$
        \item $h(\epsilon)$ has some element of the form $(q, -1)$
        \item There is some history path $p$ such that the last configuration in $p$ has index 0 and $h(p)$ has an element of the form $(q,-1)$
        \item There is some history path $p$ such that the last configuration in $p$ has index $|w| - 1$ and $h(p)$ has an element of the form $(q,1)$
        \item There is some history path $p$ such that the last configuration in $p$ is $(q,u)$ and $h(p)$ does not satisfy $\delta(q,w_u)$
    \end{enumerate}

    Condition 1 does not hold because $h(\epsilon)$ is defined to be $(q_0,0)$. Similarly, condition 2 does not hold because
    if $h(\epsilon)$ has an element of the form $(q,-1,q')$ then $(q',-1) \in f(x_r)$, where $x_r$ is the root of $\rho$,
    which implies, contra assumption, that $\rho$ is not accepting.

    Now to address conditions 3-5. Let $p$ be a history that is defined for $h$. Let $q$ be the state label of the last element of $p$.
    By the definition of $h$ we have that there is a node $x$ in $\rho$ with configuration $(q,i)$, where $0 \leq i < |w|$ and $q \in Q$.
    If condition 3 holds, then $i = 0$ and $h(p)$ has an element of the form $(q,-1)$. But this would imply that some node $x$ in $\rho$
    has some child whose index is -1, which would mean that, per Definition \ref{def:run}, $\rho$ is not a run of $A$ on $w$. But $\rho$ is an accepting run of $A$ on $w$.

    If condition 4 holds, then $i = |w| - 1$ and $h(p)$
    has an element of the form $(q,1)$.

    But this would imply that some node $x$ in $\rho$ has some child whose index is greater than $|w| - 1$, which would mean that, per Definition \ref{def:run}, $\rho$ is not a run of $A$ on $w$. But $\rho$ is an accepting run of $A$ on $w$.

    If condition 5 holds, there is some node $x$ in $\rho$ with configuration
    $(q,u)$ such that
    $\mypath(x) = p$, $h(p) = f(x)$, and $f(x)$ does not satisfy $\delta(q,w_u)$. But since $f(x)$ is the set of configuration labels of
    $x$'s children, this implies that the set of configuration labels of $x$'s children does not satisfy $\delta(q,w_u)$, which
    contradicts the assumption that $\rho$ is accepting. Since none of the conditions hold, $h$ is on $w$.

    Now to show that $h$ is accepting. Let $\beta = (q_1,u_1),(q_2,u_2),\ldots$ be a history path in $h$. If $\beta$ is finite and has length $n$, then it
    must be the case that $\delta(q_n, w_{u_n}) = \textbf{true}$, else condition 4 would be violated. So if
    $\beta$ is finite then $\beta$ is an accepting history path. Suppose that $\beta$ is infinite.
    Since $\rho$ is accepting, every
    infinite length branch in $\rho$ has infinitely many nodes labeled with configurations with an accepting state.
    Let $p$ be a history that represents a partial branch of an infinite branch of $\rho$. Now, since $p$ is finite
    it cannot have infinitely many nodes labeled with configurations with an accepting state, so there must be
    another node labeled with an accepting state in $\beta$ that is not in $p$, and there must be another history path $p'$
    that is an extension of $p$ that leads to this node. Formally, there exists a history path $p'$ such that $p$ is a subsequence
    of $p'$ and $(d,q) \in f(p')$ for some $q \in F$. So $(d,q) \in h(p')$. This implies that $q \in \inf(\beta)$ and since
    $q$ is an accepting state we have that $\beta$ is accepting. Since all history paths in $h$ are accepting we have that $h$ is a winning history strategy on $w$.

    Now, $G_{A,w}$ is a special instance of a parity game. A classic result for parity games
    is that if a parity game has a winning history strategy,
    then it has a memoryless winning strategy, i.e., it has a strategy that relies only on the current
    configuration of the game \cite{emersonandjutla}. Let
    $h' : (Q \times \mathbb{N}) \rightarrow 2^{\{-1,0,1\} \times Q}$ be a winning memoryless strategy for $G_{A,w}$.
    Without loss of generality, $h'$ is equivalent to the function $\gamma : \mathbb{N} \rightarrow 2^{Q \times \{-1,0,1\} \times Q}$.
    But $\gamma$ is an accepting strategy word over $w$.

    Suppose the Protagonist has a winning strategy $\gamma$ in the game $G_{A,w}$. From $\gamma$ we can obtain an accepting
    run $\rho$ of $A$ on $w$ using the following inductive construction.
    In the base case, let the root of $\rho$, $x_r$, have configuration $(q_0,0)$. For the inductive step,
    if $x$ is a node in $\rho$ and
    $x$ has configuration $(q,u)$, then for every configuration $c$ in the set $\{ (q',u+d) \mid (q,d,q') \in \gamma(u) \}$, we have that
    $x$ has a child with configuration $c$. Because $\gamma$ is a winning strategy, all finite strategy paths reach a $\textbf{true}$ transition,
    i.e., it is not possible to reach a $\textbf{false}$ transition. So all finite branches in $\rho$ must be accepting.
    Moreover, because $\gamma$ is a winning strategy, all infinite strategy paths transition through at least one accepting state infinitely many times.
    So all infinite branches in $\rho$ must be accepting. So $\rho$ is an accepting run, which means $A$ accepts $w$.
\end{proof}

In the ``if" direction of the proof,
an accepting run is built by assigning children to nodes according to the transitions prescribed by
the strategy word. One may wonder why, in the ``only if'' direction of the proof, we do not build a strategy word by looking at the accepting run
and add transitions to the strategy word according to the transitions that occur in the accepting run.
The reason is that, because disjunctive transitions allow for nondeterminism,
nodes in an accepting run do not necessarily make
the same transitions even if they are labeled with the same state and index.

One response to this observation might be to attempt to resolve the nondeterminism by having our
strategy word prescribe the satisfaction of all disjuncts that are satisfied somewhere in the
accepting run. The argument for this approach might go something like the following. Suppose there is is an accepting
run $\rho$ with two nodes with label $(A \vee B, n)$ and that one of the nodes has children that satisfy $A$ and the other
node has children that satisfy $B$. But since $\rho$ is an accepting run, and satisfying
either of $A$ or $B$ leads to an accepting branch, it might seem reasonable to have our strategy word prescribe
the satisfaction of both $A$ and $B$ whenever the state $A \vee B$ occurs at index $n$.

The problem with this approach is that there are some runs for which the only way to achieve an accepting
branch after satisfying $A$ is for the run to cycle back to $(A \vee B, n)$ and then to satisfy $B$.
Here is a brief example to illustrate the problem. Let $\pi = a$ and $A = (Q,\Sigma,\delta,q_0,F)$,
where $Q = \{ q_0, q_1, q_2  \}, \Sigma = \{ a \}, F = \emptyset$, and $\delta$ is defined
as follows:
\begin{itemize}
    \item $\delta(q_0,a) = (q_1,0) \vee (q_2,0)$
    \item $\delta(q_1,a) = (q_0,0)$
    \item $\delta(q_2,a) = \textbf{true}$
\end{itemize}

There are infinitely many accepting runs of $A$ on $\pi$ that sometimes transition to
$(q_1,0)$ and sometimes transition to $(q_2,0)$ from state $q_0$
and index 0. However, in a run where both disjuncts are always followed from $(q_0,0)$, there is an infinite
branch that cycles between nodes labeled with $(q_0,0)$ and $(q_1,0)$. Since this branch does not contain any nodes labeled
with an accepting state it, is rejecting.

\begin{theorem}
    A two-way alternating automaton on finite words $A = (Q,\Sigma,q_0,\delta,F)$ accepts an input word $w$ iff $A$ has a strategy word $\gamma$ over $w$ and
    an accepting annotation $\eta$ of $\gamma$.
\end{theorem}

\begin{proof}

$(\leftarrow)$ Suppose $A$ has a strategy word $\gamma$ over $w$ and an accepting annotation $\eta$ of $\gamma$. For purpose of contradiction suppose that $\gamma$ is not accepting. So there is some strategy path $\beta$ in $\gamma$ that is not accepting. By Lemma 5, all finite strategy paths in $\gamma$ are accepting, so $\beta$ is an infinite strategy path. Because $\beta$ is not accepting, it has only finitely many elements with accepting states. So after a finite prefix $\beta$ has an infinite sequence that has no accepting states. Because there are a finite number of states and indexes, some state-index pair $(j,q)$ must occur more than once. Let $c$ be the first subsequence of $\beta$ that begins with $(j,q)$ and ends at $(j,q)$ without visiting any accepting states. Clearly, $c$ is a non-accepting $(j,q)$ semi-path. By the Shortening Lemma we have that there is a non-accepting $(j,q)$ semi-path $c'$ of length 1. By the definition of non-accepting $(j,q)$ semi-path $c'_0$ must be $(j,q,0,q)$. By the closure properties for semi-paths we have that $(q,0,q) \in \eta(j)$. But since $\eta$ is accepting it has no elements of the form $(q,0,q)$. We have reached a contradiction. So $\gamma$ must be an accepting strategy word. By Theorem 2 we have that $A$ accepts $w$.

$(\rightarrow)$ Suppose $A$ accepts $w$. From Theorem 2 we have that $A$ has an accepting strategy word $\gamma$ over $w$. Consider the set of subsequences of strategy paths of $\gamma$ of length greater than 1 such that the first element is $(j,q)$ and the last element is $(j,q')$ for some $0 \leq j < |w|$ and $q,q' \in Q$; that is, the set of subsequences of strategy paths in $\gamma$ that begin and end with the same index. We call such a subsequence a $(q,q')$-cycle of $\gamma$ at index $j$. We say that such a cycle is accepting if one of the states visited along it, including the last state $q'$ but excluding the first state $q$, is an accepting state in $F$.
Let $\eta : |w| \rightarrow 2^{Q \times \{0,1\} \times Q}$ be a mapping such that for $0\leq j < |w|$ we have that $\eta(j)$ consists of all triples $(q,f,q')$ such that $(q,q')$ is a cycle of $\gamma$ at $j$, with $f=1$ if this cycle is accepting, and $f=0$ otherwise.

We claim that $\eta$ is an accepting annotation of $\gamma$. We first show that $\eta$ is an annotation. Now, $\eta$ is an annotation if it satisfies the four closure conditions for annotations:

\begin{enumerate}
\item
If $(q,0,q') \in \gamma(j)$ then $(j,q),(j,q')$ is a $(q,q')$-cycle of $\gamma$ at $j$. So $(q,\final(q'),q') \in \eta(j)$.
\item
If $(q,1,q') \in \gamma(j)$ and $(q',-1,q'') \in \gamma(j+1)$ then $(j,q),(j+1,q'),(j,q'')$ is a $(q,q'')$-cycle of $\gamma$ at $j$. So $(q,\max(\final(q'),\final(q'')),q'') \in \eta(j)$.
\item
If $(q,-1,q') \in \gamma(j)$ and $(q',1,q'') \in \gamma(j-1)$ then $(j,q),(j-1,q'),(j,q'')$ is a $(q,q'')$-cycle of $\gamma$ at $j$. So $(q,\max(\final(q'),\final(q'')),q'') \in \eta(j)$.
\item
If $(q_1,f_1,q_2) \in \eta(j)$ and $(q_2,f_2,q_3) \in \eta(j)$ then there is a $(q_1,q_2)$-cycle of $\gamma$ at $j$ and there is a $(q_2,q_3)$-cycle of $\gamma$ at $j$. It follows that there is a $(q_1,q_3)$-cycle of $\gamma$ at $j$. Moreover, if the $(q_1,q_2)$ cycle is accepting or the $(q_2,q_3)$ cycle is accepting then the $(q_1,q_3)$ cycle is accepting, otherwise the $(q_1,q_3)$ cycle is not accepting. It follows that $(q_1,\max(f_1,f_2),q_3) \in \eta(j)$.
\end{enumerate}

For the purpose of contradiction suppose that $\eta$ is not an accepting annotation, that is, for some $j$ $(q,0,q) \in \eta(j)$. Then there is a non-accepting $(q,q)$ cycle $c$ of $\gamma$ at $j$. It follows that there is an infinite strategy path $\beta$ in $\gamma$ that consists of a finite prefix $p$ followed by infinite concatenations of $c$. Because $c$ is non-accepting it has no elements with accepting states, so all visits to accepting states occur within $p$, which is finite. Because $\beta$ has only finitely many visits to accepting states, $\beta$ is not accepting. So $\gamma$ is not an accepting strategy word. But $\gamma$ is an accepting strategy word. We have reached a contradiction. So $\eta$ is an accepting annotation, as was to be shown.
\end{proof}

\begin{corollary}
Let $A = (Q, \Sigma, q_0, \delta, F)$ be a two-way alternating automaton on finite words. Then $A$ accepts an input word $w$ iff there is a word $\gamma$ of length $|w|$ over $2^{Q \times \{-1,0,1\} \times Q}$ and a word $\eta$ of length $|w|$ over $2^{Q \times \{0,1\} \times Q}$ such that the following holds:

\begin{enumerate}
    \item For all $(q,d,q') \in \gamma(0)$, we have that $q = q_0$ and $d \neq -1$
    \item For all $(q,d,q') \in \gamma(|w|-1)$, we have that $d \neq 1$
    \item For each $i$, $0 \leq i < |w|$, for each element of $\{q \mid (q,d,q') \in \gamma(i)\}$, we have that the set $\{(d,q'') \mid (q,d,q'') \in \gamma(i)\}$ satisfies $\delta(q,w_i)$
    \item If $(q,0,q') \in \gamma(i)$, then we have that $(q,\final(q'),q') \in \eta(i)$
       \item If $(q,f_1,q') \in \eta(i)$ and $(q',f_2,q'') \in \eta(i)$, we have that $(q,\max(f_1,f_2),q'') \in \eta(i)$
    \item For all $i$, $0 \leq i < |w|$, we have that $\eta(i)$ has no elements of the form $(q,0,q)$

    \item If $i_1 = i_2-1$, $(q,1,q') \in \gamma(i_1)$, and $(q',-1,q'') \in \gamma(i_2)$, we have that $(q,\max(\final(q'),\allowbreak \final(q''),q'') \in \eta(i_1)$
    \item If $i_1 = i_2+1$, $(q,-1,q') \in \gamma(i_1)$, and $(q',1,q'') \in \gamma(i_2)$, we have that $(q,\max(\final(q'),\allowbreak \final(q''),q'') \in \eta(i_1)$

\end{enumerate}

\end{corollary}
\begin{proof}
By Theorem 3 we have that $A$ accepts $w$ iff it has an accepting strategy word $\gamma$ over $w$ and an accepting annotation $\eta$ of $\gamma$. By the definitions of strategy word and accepting annotation provided earlier in this section, $\gamma$ and $\eta$ are words that meet the conditions of this corollary.
\end{proof}

\begin{theorem}
  Let $A = (Q,\Sigma,\delta,q_0,F)$ be a 2AFW and let $w$ be a word over the alphabet $\Sigma=2^{AP}$. There is a deterministic finite automaton (DFA) $A_d = (S,T,\Delta,s_0,G)$, with alphabet $T = 2^{AP} \times 2^{Q \times \{-1,0,1\} \times Q} \times 2^{Q \times \{0,1\} \times Q}$ such that $A$ accepts $w$ iff there are a word $\gamma$ of length $|w|$ over the alphabet $2^{Q \times \{-1,0,1\}\times Q}$ and a word $\eta$ also of length $|w|$ over the alphabet $2^{Q \times \{0,1\}\times Q}$ and $A_d$ accepts the word $(w_0, \gamma_0, \eta_0), (w_1, \gamma_1, \eta_1), \ldots, (w_{|w|-1}, \gamma_{|w|-1}, \eta_{|w|-1})$.

\end{theorem}

\begin{proof}
  Our proof proceeds according to the following. First, we describe a DFA $A_1 = (S_1, T, \Delta_1, s^1_0, G_1)$ that accepts $w$ iff conditions 1-6 of Corollary 3.1 hold. Checking these conditions requires constant memory. Then we describe a DFA $A_2 = (S_2, T, \Delta_2, s^2_0, G_2)$ that accepts $w$ iff conditions 7-8 of Corollary 3.1 hold. Checking these conditions requires at every step remembering the previous input symbol, which requires a number of states exponential in the size of $A$. We take the intersection of $A_1$ and $A_2$ to achieve $A_d$, which satisfies the conditions of the theorem.

Here is the construction for $A_1$. Let $S_1 = \{s^1_0,s^1_a,s^1_{a'},s^1_r\}$, $T = 2^{AP} \times 2^{Q \times \{-1,0,1\} \times Q} \times 2^{Q \times \{0,1\} \times Q}$, $G_1 = \{s^1_a\}$. Intuitively, $s^1_0$ is the initial state; the state $s^1_a$ indicates that the input read so far has been acceptable and the automaton is ready to accept; the state $s^1_{a'}$ indicates that the input read so far has been acceptable, but the most recent move prescribed by $\gamma$ involves moving to the right, so the automaton is not ready to accept because if the last input character has just been read this indicates that $\gamma$ prescribes moving to the right, which is out of bounds of $w$; finally, the state $s^1_r$ is a rejecting sink.

$\Delta_1$ is defined according to the following:
\begin{enumerate}
\item $\Delta_1(s^1_a,(I,M,E)) = s^1_a$ if all of the following hold:
\begin{enumerate}
    \item For each element $q$ of $\{q \mid (q,d,q') \in M\}$ we have that the set $D=\{(d,q') \mid (q,d,q') \in M\}$ satisfies $\delta(q,I)$ and
    for all $(d,q')\in D$ we have $d \neq 1$.
    \item For all $(q,0,q') \in M$ we have that $(q,\final(q'),q') \in E$
    \item For each pair $(q,f_1,q'),(q',f_2,q'') \in E$ we have that $(q,\max(f_1,f_2),q'') \in E$
    \item $E$ contains no elements of the form $(q,0,q)$
\end{enumerate}
\item $\Delta_1(s^1_a,(I,M,E)) = s^1_{a'}$ if conditions 1b-1d hold and for each element of $\{q \mid (q,d,q') \in M\}$ we have that the set $D = \{(d,q') \mid (q,d,q') \in M\}$ satisfies $\delta(q,I)$ and for some element $q$ of $\{q \mid (q,d,q') \in M\}$ we have a triple $(q,1,q')\in M$.
\item $\Delta_1(s^1_a,(I,M,E)) = s^1_r$ if neither (1) nor (2) holds.
\item $\Delta_1(s^1_0,(I,M,E)) = s^1_a$ if conditions 1b-1d hold, and for each element $q$ of $\{q \mid (q,d,q') \in M\}$ we have that the set $D = \{(d,q') \mid (q,d,q') \in M\}$ satisfies $\delta(q,I)$ and for all $(d,q') \in D$ we have that $d = 0$.

\item $\Delta_1(s^1_0,(I,M,E)) = s^1_{a'}$ if conditions 1b-1d hold, for all elements $d$ of $\{d \mid (q,d,q') \in M\}$ we have that $d \neq -1$, and for some element $q$ of $\{q \mid (q,d,q') \in M\}$ we have that the set $D=\{(d,q') \mid (q,d,q') \in M\}$ satisfies $\delta(q,I)$ and $d = 1$
for some $(d,q')\in D$.
\item $\Delta_1(s^1_0,(I,M,E)) = s^1_r$ if neither (4) nor (5) holds.
\item $\Delta_1(s^1_{a'},(I,M,E)) = s^1_a$ if conditions 1a-1d hold
\item $\Delta_1(s^1_{a'},(I,M,E)) = s^1_{a'}$ if condition 2 holds
\item $\Delta_1(s^1_{a'},(I,M,E)) = s^1_r$ neither (1) nor (2) holds.
\end{enumerate}

Now for the construction of $A_2$. The state set is $S_2 = 2^{Q \times \{-1,0,1\} \times Q} \times 2^{Q \times \{0,1\} \times Q} \cup s^2_r$, the alphabet is $T = 2^{AP} \times 2^{Q \times \{-1,0,1\} \times Q} \times 2^{Q \times \{0,1\} \times Q}$, the accepting set $G_2 = S_2 \setminus s^2_r$, $s^2_0 = \gamma(0)$. Intuitively, the states of $A_2$ remember the most previous strategy word and annotation elements. This allows the DFA to check that conditions 7-8 of Corollary 3.1 hold; if not, then the DFA transitions to a rejecting sink. If the conditions hold, then the DFA transitions to the state corresponding to the current element of $\gamma$ and the evaluation proceeds. All states are accepting except for the rejecting sink, so if conditions 7-8 hold for all elements of $\gamma$ and $\eta$ then the DFA will be in an accepting state after consuming all input.

The transition function $\Delta_2$ is defined according to the following:
\begin{enumerate}
\item For all $(M,E) \in S_2$,
$\Delta_2((M,E),(I,M',E')) = (M',E')$ if the following conditions hold:
\begin{enumerate}
\item
If $(q,1,q') \in M$, $(q',-1,q'') \in M'$, then $(q,\max(\final(q'),\final(q''),q'') \in E$, and

\item
if $(q,-1,q') \in M'$, $(q',1,q'') \in M$,
then $(q,\max(\final(q'),\final(q'')),q'') \in E'$

\end{enumerate}
If either condition fails to hold then $\Delta_2((M,E),(I,M',E')) = s^2_r$.

\item $\Delta_2(s^2_r,(I,M,E))=s^2_r$
    \end{enumerate}

We take the intersection of $A_1$ and $A_2$ to get $A_d$, which satisfies the conditions of the theorem.
\end{proof}

\begin{corollary}
Let $A= (Q,\Sigma,\delta,q_0,F)$ be a 2AFW and let $w$ be a word over the alphabet $\Sigma = 2^{AP}$. There is an NFA $A_n$
such that $A_n$ accepts $w$ iff $A$ accepts $w$, $A_n$ has $2^{O(|Q|^2)}$ states, and each state of $A_n$ is of size

$O(|Q|^4)$.
\end{corollary}
\begin{proof}

By Theorem 4 we have that there is a DFA $A_d = (S,\Sigma \times 2^{Q \times \{-1,0,1\} \times Q} \times 2^{Q \times \{0,1\} \times Q},\Delta,\allowbreak s_0,G)$ and that $A$ accepts $w$ iff there are a word $\gamma$ of length $|w|$ over $2^{Q \times \{-1,0,1\} \times Q}$ and a word $\eta$ of length $|w|$ over $2^{Q \times \{0,1\} \times Q}$ and $A_d$ accepts the word $(w_0, \gamma_0, \eta_0),(w_1,\gamma_1,\eta_1), \ldots,\allowbreak (w_{|w|-1},\gamma_{|w|-1},\eta_{|w|-1})$.

To construct an NFA $A_n$ that accepts $w$ iff $A$ accepts $w$, we have $A_n$ \emph{guess} $\gamma$ and $\eta$ and then simulate $A_d$ over $(w_0, \gamma_0, \eta_0),(w_1,\gamma_1,\eta_1),\ldots,(w_{|w|-1}, \gamma_{|w|-1},\eta_{|w|-1})$.
Formally, NFA $A_n = (S',\Sigma, \Delta',s'_0,G')$ is as follows:
\begin{itemize}
    \item $S' = S \times 2^{Q \times \{-1,0,1\} \times Q} \times 2^{Q \times \{0,1\} \times Q}$

    \item $s'_0 = s_0 \times 2^{Q \times \{-1,0,1\} \times Q} \times 2^{Q \times \{0,1\} \times Q}$

    \item $G' = G \times 2^{Q \times \{-1,0,1\} \times Q} \times 2^{Q \times \{0,1\} \times Q}$

    \item $\Delta'((s,m,e),i) = \Delta(s,(i,m,e)) \times 2^{Q \times \{-1,0,1\} \times Q} \times 2^{Q \times \{0,1\} \times Q}$
\end{itemize}

To determine the size of $S' = S \times 2^{Q \times \{-1,0,1\} \times Q} \times 2^{Q \times \{0,1\} \times Q}$, we recall that $S = S_1 \times S_2$, where $S_1$ is of constant size, $S_2 = 2^{Q \times \{-1,0,1\} \times Q} \times 2^{Q \times \{0,1\} \times Q} \cup s^2_r$. So the overall size of $S'$ is $2^{O(|Q|^2)}$.

To determine the size of a state in $S'$, we note that an element in $S'$ is a tuple consisting of (a) an element of $S_1$, which is a fixed-size set, (b) an element of $S_2$ is is a pair consisting of a subset of $Q \times \{-1,0,1\} \times Q$ and a subset of  $Q \times \{0,1\} \times Q$, (c) a subset of $Q \times \{-1,0,1\} \times Q$, and (d) a subset of $Q \times \{0,1\} \times Q$. So the size of each state is $O(|Q^2|)$.
\end{proof}

\begin{theorem}
Let $\varphi$ be an ALDL$_f$ formula. Satisfiability checking of $\varphi$ is PSPACE-complete.
\end{theorem}

\begin{proof}
To establish the lower bound, we note that the satisfiability problem for LTL$_f$ is PSPACE-complete and that a formula of LTL$_f$ can be converted to ALDL$_f$ in linear time and space using the following procedure: translate from LTL$_f$ to LDL$_f$ using the linear construction provided in \cite{ldlf}, then translate from LDL$_f$ to ALDL$_f$ by converting each instance of path expression to a path automaton using Thompson's construction \cite{thompson}.

Now for the upper bound. Definition 10 provides the construction of a 2AFW $A_{\varphi} = (Q, \Sigma, \delta, q_0, F)$ from an ALDL$_f$ formula $\varphi$. By Lemma 2 we have that $|Q| \leq |\varphi|$ because $Q$ is the Fischer-Ladner closure of $\varphi$. By Theorem 1 we have that $\mathcal{L}(\varphi) = \mathcal{L}(A_\varphi)$.

By Corollary 4.1, we have that there is an NFA $A_n$ such that $\mathcal{L}(A_{\varphi}) = \mathcal{L}(A_n)$ and that $A_n$ has $2^{O(|Q|^2)}$ states and that each state is of size $O(|Q|^4)$. Reachability from the initial state of $A_n$ to an accepting state of $A_n$ can be performed through a nondeterministic search: at each step, guess a letter of the alphabet and guess the successor state, while remembering only the successor state, and if an accepting state is reached as the final input is consumed then it is reachable from the initial state. Now, each letter is polynomial in the size of $\varphi$ and by Corollary 4.1 we have that each state of $A_n$ is polynomial in the size of $\varphi$. So this procedure can be done in NPSPACE. Since $\varphi$ and $A_n$ define the same language, $\varphi$ is satisfiable iff an accepting state of $A_n$ is reachable from the initial state of $A_n$. So satisfiability of $\varphi$ is in NPSPACE. But by Savitch's Theorem \cite{savitch}, we have that NPSPACE = PSPACE. So satisfiability checking of $\varphi$ is in PSPACE.
\end{proof}

\section{Conclusion}

LDL$_f$'s use of regular-expression operators allows it to achieve greater expressiveness than LTL$_f$, yet satisfiability checking of LDL$_f$ formulas is PSPACE-complete, just like LTL$_f$. That is,  LDL$_f$ , compared to LTL$_f$, is more expressive yet no more expensive. Thus, LDL$_f$ may be an attractive alternative for many applications.

ALDL$_f$ extends the paradigm of LDL$_f$ by allowing the use of NFAs in place of regular expressions to express temporal constraints and also by providing for the direct expression of past modalities, which can be combined with present and future modalities within the same formula. Because NFAs are stateful, they can sometimes be more convenient to use than regular expressions. They are also exponentially more succinct than regular expressions. Satisfiability checking of ALDL$_f$ formulas is still PSPACE-complete, so these features come at no additional cost compared to LDL$_f$.

\bibliography{cav}
\bibliographystyle{alphaurl}

\end{document}